\documentclass[a4paper,USenglish]{lipics-v2021}
\pdfoutput=1
\nolinenumbers
\hideLIPIcs
\usepackage[noend]{algpseudocode}
\usepackage[many]{tcolorbox}
\usepackage{mathtools}
\usepackage{bm}

\usepackage{silence}
\makeatletter \newcommand{\apxthm}[3]{
\newcounter{sec#2} \newcounter{val#2} \newcounter{use#2}
\setcounter{sec#2}{\value{section}} \setcounter{val#2}{\value{theorem}}
\global\expandafter\let\csname apx@thesec@#2\endcsname\thesection
\begin{#1} \label{#2} {#3} \end{#1}
\long\@namedef{#2}{
    \let\current@thesection\thesection
    \@tempcnta=\value{section} \@tempcntb=\value{theorem}    \expandafter\let\expandafter\thesection\csname apx@thesec@#2\endcsname
    \setcounter{section}{\value{sec#2}} \setcounter{theorem}{\value{val#2}}
    \phantomsection \label{#2*} \begin{#1} {#3} \end{#1}
    \setcounter{section}{\@tempcnta} \setcounter{theorem}{\@tempcntb}
    \let\thesection\current@thesection
}}

\newenvironment{algobox}[2][]{
    \begin{center}
    \begin{tcolorbox}[enhanced,title=\centering \large {#2},colback=white,colframe=black!80,width=\textwidth]
    \underline{\textbf{Code for a party $i$}{#1}}
    \begin{algorithmic}[1]
}{
    \end{algorithmic} 
    \end{tcolorbox} 
    \end{center}
}

\makeatletter\newif\ifTightALG\TightALGfalse
\patchcmd{\ALG@doentity}{\item[]\nointerlineskip}{\ifTightALG\else\item[]\nointerlineskip\fi}{}{}
\def\wALGone#1{\expandafter\let\csname old#1\expandafter\endcsname\csname #1\endcsname
  \expandafter\def\csname #1\endcsname##1{\TightALGtrue\csname old#1\endcsname{##1}\TightALGfalse}}
\def\wALGzero#1{\expandafter\let\csname old#1\expandafter\endcsname\csname #1\endcsname
  \expandafter\def\csname #1\endcsname{\TightALGtrue\csname old#1\endcsname\TightALGfalse}}
\wALGone{If} \wALGone{For}
\wALGzero{Else} \wALGzero{EndIf} \wALGzero{EndFor}
\makeatother

\algblock{Upon}{EndUpon} \algrenewtext{Upon}[1]{\textbf{upon} {#1} \textbf{do}}
\makeatletter \ifthenelse{\equal{\ALG@noend}{t}} {\algtext*{EndUpon}} \makeatother

\DeclarePairedDelimiter{\ceil}{\lceil}{\rceil}
\newcommand{\msg}[1]{\langle {#1} \rangle}

\newcommand{\supp}{\mathsf{supp}}
\newcommand{\pj}{\mathsf{proj}}
\newcommand{\tpj}{\mathsf{2\textsf-proj}}
\newcommand{\E}{\mathbb{E}}
\newcommand{\Z}{\mathcal{Z}}
\newcommand{\Q}{\mathcal{Q}}
\renewcommand{\H}{\mathcal{H}}

\renewcommand{\S}{\mathcal{S}}
\newcommand{\R}{\mathcal{R}}
\newcommand{\V}{\mathcal{V}}
\newcommand{\W}{\mathcal{W}}
\newcommand{\val}{\mathsf{val}}
\newcommand{\REQ}{\texttt{REQ}}
\newcommand{\SYM}{\texttt{SYM}}
\newcommand{\PullCast}{\mathsf{PullCast}}
\newcommand{\Term}{\mathsf{Term}}

\title{Multivalued Consensus: General Adversaries Require More Communication}

\author{Mose Mizrahi}{ETH Zurich, Switzerland}{mmizrahi@ethz.ch}{https://orcid.org/0009-0009-9771-0845}{}
\author{Roger Wattenhofer}{ETH Zurich, Switzerland}{wattenhofer@ethz.ch}{https://orcid.org/0000-0002-6339-3134}{}
\authorrunning{M.\ Mizrahi and R.\ Wattenhofer}
\Copyright{Mose Mizrahi Erbes, and Roger Wattenhofer}
\EventEditors{Ioannis Chatzigiannakis, Andrea Vitaletti, Keren Censor-Hillel, and William K. Moses Jr.}
\EventNoEds{4}
\EventLongTitle{40th International Symposium on Distributed Computing (DISC 2026)}
\EventShortTitle{DISC 2026}
\EventAcronym{DISC}
\EventYear{2026}
\EventDate{November 9--13, 2026}
\EventLocation{Rome, Italy}
\EventLogo{}
\SeriesVolume{397}
\ArticleNo{43}

\begin{document}

\maketitle

\begin{abstract}

We study $n$-party fault-tolerant consensus against general (non-threshold) adversaries, which are characterized by adversary structures that enumerate every party set $Z$ such that the adversary can corrupt every party in $Z$. An adversary structure $\Z$ satisfies the $\Q^d$ condition if the complete party set $[n]$ is not the union of $d$ or fewer sets in $\Z$. Adversaries characterized by $\Q^d$-satisfying adversary structures generalize threshold adversaries which can corrupt at most $t < \frac{n}{d}$ parties.

For every $d \geq 1$, we describe an infinite family $\Z_\pj^{n,d}$ of $\Q^d$-satisfying $n$-party adversary structures based on finite projective geometry such that the byzantine adversaries characterized by this family cause error-free $R$-round protocols for interactive consistency on $L$-bit inputs (for agreement on every party's $L$-bit input) to require $\Omega(Ln^{2+1/d})$ bits of expected communication, if $L = \Omega(Rn^{-1/d} + 1)$. Likewise, $\Z_\pj^{n,d}$ causes byzantine agreement and broadcast to cost $\Omega(Ln^{1+1/d})$ bits, if $L = \Omega(Rn^{1-1/d})$. In every case, the lower bound is $\Omega(L_{\mathsf{out}} \cdot n^{1+1/d})$ bits, where $L_{\mathsf{out}}$ is the length of the output. This is in contrast to the case of threshold adversaries, for which the baseline $\Omega(L_{\mathsf{out}} \cdot n)$ ``every party must learn the output'' lower bound is tight for all large $L_{\mathsf{out}}$ if the adversary can corrupt $\Theta(n)$ parties.

The adversary structure $\Z_\pj^{n,d}$ also makes $\Omega(Ln^{1+1/d})$ bits of communication required for reliable broadcast and byzantine agreement in asynchronous networks, when $L = \Omega(n^{1-1/d})$. Similarly, there is a family $\Z_\tpj^{n,d}$ of $\Q^d$-satisfying adversary structures that cause asynchronous core set agreement to \linebreak cost $\Omega(Ln^{2+1/d})$ bits. These asynchronous complexity bounds hold against send-omission adversaries, even if the protocol uses cryptography. Their basis is that if a quorum of non-faulty parties agree on an output and terminate, then the messages they sent before terminating must suffice for the parties outside the quorum to also terminate with the same output. Surprisingly, if we do not require the parties to terminate (stop sending messages) after they output, then these lower bounds no longer hold. We show this by designing a non-terminating reliable broadcast protocol for general-omission adversaries that can for any parameter $\delta > 1$ be tuned to cost $(1 + \frac{1}{\delta - 1})Ln + O(\delta n^2\log(\delta n))$ bits, which is of independent interest. Lastly, we show how to obtain termination with $O(Ln^{1+1/d} + n^2\log n)$ bits (assuming the $\Q^d$ condition), and thus prove our asynchronous lower bounds tight for large $L$.

\end{abstract}

\section{Introduction}

We study agreement in a network of $n$ message-passing parties, where an adversary causes some of the parties to fail. We show that against general $\Q^d$ adversaries \cite{hm00}, agreement tasks often require $\Omega(n^{1/d})$ times more communication than they do against threshold adversaries. For example, with $L$-bit inputs, error-free (perfectly secure) byzantine broadcast \cite{lsp82} costs $\Omega(Ln^{1+1/d})$ bits. This exceeds the baseline $\Omega(Ln)$ ``everyone must learn the sender's input'' lower bound, which is tight for all sufficiently large $L$ against threshold adversaries \cite{chen21}.

The standard way to characterize the fault tolerance of an agreement protocol is to check how many faulty parties it can tolerate. That is, one considers a $t$-threshold adversary which can corrupt at most $t$ parties, and asks if the protocol can tolerate all such adversaries.

The optimal threshold fault tolerance for an agreement task is usually of the form $t < \frac{n}{d}$, where $d \geq 1$ is some integer. For example, the optimal fault tolerance for byzantine agreement (BA) \cite{lsp82} is $t < \frac{n}{2}$ or $t < \frac{n}{3}$, depending on the setting \cite{lsp82,ds83,toueg}. The reason is that when $t \geq \frac{n}{d}$, there are $d$ sets $Z_1, \dots, Z_d$ whose union is the full party set $\{1,\dots,n\} = [n]$ such that any of the sets $Z_i$ could be the set of corrupt parties. For BA, such $d$ sets must not exist.

The literature captures this notion with (general) adversary structures \cite{hm00}. An adversary structure $\Z$ is a down-closed family of subsets of $[n]$, and $\Z$ satisfies the $\Q^d$ condition if no $d$ sets in it cover $[n]$; that is, if $\bigcup_{i = 1}^d Z_i \neq [n]$ for all $Z_1, \dots, Z_d \in \Z$. We define a $\Z$-adversary to be one characterized by $\Z$, with the ability to corrupt a set of parties $Z$ iff $Z \in \Z$.

For some intuition as to why the $\Q^d$ condition is important, we consider a very simple synchronous protocol where 1) the parties (who all have inputs) send their inputs to everyone, and 2) each party $i$ outputs $y_i$ if it has received $y_i$ from every party except for a set of parties $S_i \in \Z$ (and outputs some default value $\bot$ if there exists no such $y_i$). If $\Z$ satisfies the $\Q^3$ condition, then the following holds: If any two correct (non-faulty) parties $i$ and $j$ respectively \linebreak output $y_i \neq \bot$ and $y_j \neq \bot$, then $y_i = y_j$. To see why, consider the following three sets in $\Z$.
\begin{itemize}
    \item $S_i$: the set of parties who did not send $i$ the input $y_i$. This set is in $\Z$ since $i$ output $y_i$.
    \item $S_j$: the set of parties who did not send $j$ the input $y_j$. This set is in $\Z$ since $j$ output $y_j$.
    \item $F$: the set of faulty parties, also in $\Z$.
\end{itemize}
By the $\Q^3$ condition, there exists a party $k$ outside $S_i \cup S_j \cup F$. This is a correct party who sent $y_i$ to $i$ and $y_j$ to $j$, which means that both $y_i$ and $y_j$ are equal to the input of party $k$.

The protocol described can also be shown to guarantee that if the correct parties have a common input $v$, then they all output $v$; hence, it solves a task known as crusader agreement \cite{d82,aw24}. It is a translation of a threshold-based protocol to the general adversary setting, where the change is that instead of checking if $|S_i| < \frac{n}{3}$, the parties check if $S_i \in \Z$. Translations like this are often possible, but they are not always so simple or efficient.

In this work, we study five fundamental agreement tasks: byzantine agreement \& broadcast \cite{lsp82}, interactive consistency \cite{lsp82}, reliable broadcast \cite{b87} and core set agreement \cite{bcg93}. For all of these tasks, feasibility against general byzantine adversaries is well understood. Authenticated byzantine broadcast (and thus interactive consistency) is always possible thanks to the Dolev-Strong protocol \cite{ds83}. For the remaining tasks, there is a correspondence: Each task we study optimally requires the $\Q^d$ condition in the general adversary setting if and only if it optimally requires $t < \frac{n}{d}$ in the threshold adversary setting. This correspondence has been established in a series of works, for byzantine and omission faults \cite{hm00, kf05, z10, c23, bz25, acc25}. Intuitively, a protocol against $t < \frac{n}{d}$ faults can usually be translated into a protocol against $\Q^d$ adversaries, and the impossibility of solving a task against $t \geq \frac{n}{d}$ faults can usually be translated into the impossibility of solving the task against non-$\Q^d$ adversaries.

While feasibility is well-understood, efficiency is not. General adversary tolerant agreement protocols are often less efficient. For example, synchronous binary byzantine agreement admits a deterministic error-free solution against $t < \frac{n}{3}$ that costs $O(n^2)$ bits of communication \cite{bgp92} based on a recursive divide-and-conquer construction which only works well against threshold adversaries, while against $\Q^3$ adversaries the most efficient deterministic error-free protocol we know of is the $O(n^3)$-bit one in \cite{bz25} based on the phase king paradigm \cite{bgp89}. Hence, we ask: \emph{Does agreement against general adversaries require more communication?} 

We are especially interested in how efficient general adversary tolerant agreement can be when the parties have long inputs. This question has not been studied much in the literature, perhaps since the traditional erasure/error correcting code techniques used to support long inputs efficiently (see for example \cite{fh06,nayak20,chen21,long22}) only work well against threshold adversaries. For instance, against $t < \frac{n}{3}$ omission faults, if we have a publicly known set $S$ of $n-t$ parties (up to $t$ faulty) who know of a common value $v$ and wish to inform the remaining $t$ parties of this value, then the following would be an efficient solution: \begin{enumerate}
    \item The parties in $S$ encode $v$ into $n-t$ symbols $s_1,\dots,s_{n-t}$ with an $(n-t, n-2t)$-erasure code, such that any $n-2t$ of these symbols can be decoded into $v$. With Reed-Solomon erasure coding \cite{reed-solomon}, each symbol is $\frac{L}{n-2t} + O(\log n) = O(\frac{L}{n} + \log n)$ bits long.
    \item For all $(i,j) \in [n-t] \times [t]$, the $i^\text{th}$ party in $S$ sends the symbol $s_i$ to the $j^{\text{th}}$ party outside $S$. In total, this costs $t(n-t) \cdot O(\frac{L}{n} + \log n) = O(Ln + n^2\log n)$ bits.
    \item Each party outside $S$ receives at least $n - 2t$ symbols, which is enough for it to learn $v$.
\end{enumerate}
If one translated this protocol to the general adversary setting, then the analogue of the value $n - 2t$ would be ``the minimum $k$ such that $S$ is guaranteed to contain at least $k$ correct parties.'' But unlike in the threshold adversary setting, where we always have $k = n - 2t > \frac{n}{3}$, we could have $k = O(1)$ in the general adversary setting. Then, using an $(|S|, k)$-erasure code would result in symbols of size $\frac{L}{k} + O(\log n)$ and thus a total communication complexity of $O(|S|(n-|S|)(\frac{L}{k} + \log n)) = O(Ln^2 + n^2\log n)$ bits instead of $O(Ln + n^2\log n)$. This begs the question of whether general adversaries fundamentally cause such cost increases.

Consider for instance byzantine broadcast, a standard primitive for synchronous networks where a sender party broadcasts its input in a way which guarantees that even if the sender is faulty, every correct party agrees on what the sender broadcast. If the sender has an $L$-bit input, then a byzantine broadcast protocol must (ignoring a technicality)\footnote{The technicality we are ignoring is the assumption that each party must receive $\Omega(L)$ bits in expectation to learn the sender's $L$-bit input. This is only true under certain conditions.} involve $\Omega(Ln)$ bits \linebreak of communication since each non-sender party must learn the sender's $L$-bit input.

Against byzantine threshold adversaries that can corrupt $O(n)$ parties, the ``$O(n)$ parties might have to learn an $L$-bit output'' $\Omega(Ln)$ lower bound is tight for all large $L$: See \cite{chen21,nayak20,aa25} for byzantine agreement\footnote{The bound for byzantine agreement is actually $\Omega(Lt)$ \cite{fh06}, but this simplifies to $\Omega(Ln)$ when $t = \Theta(n)$.} and broadcast, and \cite{long22} for reliable broadcast. There are some tasks (interactive consistency and core set agreement) which require $\Omega(Ln^2)$ bits of communication when the parties have $L$-bit inputs, but this is simply due to the fact that these tasks require each correct party must learn $O(n)$ parties' $L$-bit inputs. In other words, the lower bound is still $\Omega(L_\mathsf{out} \cdot n)$, where $L_\mathsf{out}$ is the size of the output that every party must learn. We show that some $\Q^d$ adversaries cause this lower bound to go up to $\Omega(L_\mathsf{out} \cdot n^{1+1/d})$. In other words, $\Q^d$ adversaries can fundamentally cause multivalued agreement to cost $\Omega(n^{1/d})$ times more bits of communication than threshold adversaries can.

While we primarily study general adversaries, we also have results for reliable broadcast with $L$-bit inputs against up to $t < n$ omission faults. We show that solving this task requires $\Omega(Ln^2)$ bits of communication if the parties must terminate (stop sending messages) after they output, which means that the echo-based protocol in \cite{ht94} is optimal, but the bit complexity $(1 + \frac{1}{\delta - 1})Ln + O(\delta n^2\log(\delta n))$ is feasible for any $\delta > 1$ if the parties are allowed to run forever. \linebreak This protocol shows that Locher's $(1.5 - o(1))Ln$ bit communication complexity lower bound for a class of reliable broadcast protocols \cite{locher24} does not extend to all non-terminating protocols.

\subsection{Contributions}

We show that against $\Q^d$ general adversaries, one can often improve the baseline $\Omega(L_\mathsf{out} \cdot n)$ lower bound to $\Omega(L_\mathsf{out} \cdot n^{1+1/d})$. We consider five very traditional and fundamental agreement tasks, defined below. Note that each of these tasks is characterized by some \emph{Validity} property which relates the parties' inputs with their outputs, the \emph{Agreement} property which requires the correct parties' outputs to match, and finally some \emph{Liveness} property which (except for reliable broadcast) requires every correct party to output from the protocol.

\subparagraph{Byzantine Agreement.} In byzantine agreement \cite{lsp82}, each party has an input, and the correct parties must agree on an output. Validity requires that if the correct parties have a common input, then they agree on it.

\subparagraph{Byzantine Broadcast.} In byzantine broadcast \cite{lsp82}, there is a publicly known sender party with an input, and the correct parties wish to agree on what they believe the sender's input is; correctly so if the sender is correct. That is, the correct parties must agree on an output $y$, and this output must be the sender's input if the sender is correct.

\subparagraph{Interactive Consistency.} In interactive consistency \cite{lsp82}, each party has an input, and the correct parties must agree on what they believe each party's input is. That is, they must agree on a vector $Y = (y_1, \dots, y_n)$ such that if a party $i$ is correct, then $y_i$ is its input. Interactive consistency is the parallel version of byzantine broadcast where everyone broadcasts an input.

\subparagraph{Reliable Broadcast.} Reliable Broadcast \cite{b87} is a weakening of byzantine broadcast tailored for asynchronous networks, where byzantine broadcast is impossible because the parties cannot \linebreak afford to wait indefinitely to learn the input of the sender whose messages might be arbitrarily delayed. Like byzantine broadcast, reliable broadcast requires that the correct parties' outputs match (agreement) and that they equal the sender's input when the sender is correct (validity). However, liveness for reliable broadcast only requires the correct parties to all output if either some correct party outputs (totality) or if the sender is correct.\footnote{In the literature, the guarantee that the correct parties all output if the sender is correct is sometimes bundled into validity. We separate this liveness guarantee from validity since validity being a pure safety property works better for our lower bounds.}

\subparagraph{Core Set Agreement.} Core Set agreement \cite{bcg93} is a weaker version of interactive consistency tailored for asynchronous networks, where interactive consistency is impossible due to the impossibility of waiting. Instead of a full vector that contains every party's input, each party $i$ must for a large core set of parties $C = \{c_1, \dots, c_{|C|}\}$ (whose complement is in the adversary structure $\Z$ characterizing the adversary) output a set $Y = \{(c_1, y_1), \dots, (c_{|C|}, y_{|C|})\}$ thus indicating its belief that each party $c_i \in C$ has the input $y_i$. Validity for core set agreement requires that no correct party's output set $Y$ contain any incorrect (sender, input) tuple for any correct sender $j$ (a tuple $(j, y)$ such that $j$ is correct but its input is not $y$), and agreement \linebreak requires both the core set $C$ and the output set $Y$ to be the same for every correct party.

\subparagraph{The Strength of the Guarantees.} Against omission faults, one can often achieve stronger guarantees than what we list above. For example, since an omission-faulty sender never lies about its input, a reliable broadcast protocol for omission faults can ensure that both the correct and the faulty parties can only output the sender's input from the protocol, no matter if the sender is correct or not. In Section \ref{pullcastsection}, we give a general-omission tolerant reliable broadcast protocol that achieves this guarantee. However, our lower bounds for send-omission \linebreak adversaries (the ones for the asynchronous setting) only require the definitions listed above.

\subparagraph{Lower Bounds.} We have five communication complexity lower bounds that use two adversary structure families $\Z_\pj^{n,d}$ and $\Z_\tpj^{n,d}$, which are based on $d$-dimensional finite projective geometry \cite{beutel}. They are parametrized by the number of parties $n$ and the least $d \geq 1$ such that the $\Q^d$ property is satisfied. For every $d \geq 1$, each one of $\Z_\pj^{n,d}$ and $\Z_\tpj^{n,d}$ exists for infinitely many $n$. \linebreak Hence, for every $d \geq 1$, our lower bounds hold for infinitely many $n$.\footnote{One can further prove that for every $d \geq 1$ there is some $N_d$ such that our bounds hold for all $n \geq N_d$. This is because our lower bounds monotonically grow with $n$, and because the ratio of the consecutive values of $n$ for which our adversary structures exist approaches $1$ as $n \rightarrow \infty$. We provide the details of this argument in the \hyperref[infinitely-many]{appendix}.} Our results hold for all $d \geq 1$; however, to prevent our lower bounds from being vacuous, we state them with the requirement ($d \geq 1$, $d \geq 2$ or $d \geq 3$) which makes the studied task possible.

We use expected communication complexity in our lower bounds. This expectation is over the parties' inputs (which follows a distribution of our choice), plus the parties' randomnesses assuming the protocol is randomized. We say that a protocol must involve $\Omega(f(L,n,d))$ bits of expected communication if there is an input distribution which induces this requirement. 

\apxthm{theorem}{interthm}{
    For all $d \geq 3$, all $R \geq 1$ and all $L = \Omega(Rn^{-1/d} + 1)$, any $R$-round randomized error-free synchronous interactive consistency protocol for $L$-bit inputs that can tolerate all byzantine $\Z_\pj^{n,d}$-adversaries must involve $\Omega(Ln^{2+1/d})$ bits of expected communication.
}

Theorem \ref{interthm} also implies the following lower bound for byzantine agreement and broadcast.

\apxthm{corollary}{babbcorr}{
    For all $d \geq 3$, $R \geq 1$ and $L = \Omega(Rn^{1-1/d})$, any $R$-round randomized error-free synchronous byzantine agreement or broadcast protocol for $L$-bit inputs that can tolerate all byzantine $\Z_\pj^{n,d}$-adversaries must involve $\Omega(Ln^{1+1/d})$ bits of expected communication.
}

Note that the restriction to error-free protocols is essential. By \cite{co18}, $O(Ln + \kappa \cdot \mathit{poly}(n))$ bits and $\mathit{poly}(n)$ rounds suffice for secure byzantine broadcast against any adversary structure if the protocol is allowed to fail with a probability negligible in $\kappa$, which means that Corollary \ref{babbcorr} (and thus Theorem \ref{interthm}) does not hold without the restriction to error-free protocols. Intuitively, this is because the bottleneck that causes the lower bounds to hold in the synchronous setting is the correct parties checking if their inputs (or what they believe the broadcaster's input is) are the same, and input hashing makes this cheap outside the error-free setting.

For byzantine broadcast, we get the $\Omega(Ln^{1+1/d})$ lower bound by just reducing interactive consistency to $n$ parallel instances of byzantine broadcast (one for each party to broadcast its input), and thus showing that interactive consistency is at most $n$ times more expensive than byzantine broadcast. In turn, the lower bound for byzantine broadcast implies the lower bound for byzantine agreement since one can reduce byzantine broadcast to byzantine agreement on \linebreak what the sender's input is, preceded by a single round where the sender sends its input to the other parties with $L(n-1) = o(Ln^{1+1/d})$ bits. We discuss the technicalities in Section \ref{syncsection}.

Our next lower bounds are for asynchronous protocols. They only require send-omission faults, and hold even if the protocol can fail to achieve output safety (agreement and validity) with a constant probability. Moreover, these lower bounds hold even if the adversary must be efficient and explicitly constructed, which means that one cannot circumvent them by using cryptography. However, the protocol must be terminating. That is, liveness should hold with \linebreak probability $1$, and the correct parties must terminate (stop sending messages) upon outputting.

\apxthm{theorem}{rbthm}{
    For all $d \geq 1$ and $L = \Omega(n^{1-1/d})$, a terminating reliable broadcast protocol for $L$-bit inputs that can tolerate all efficient send-omission $\Z_\pj^{n,d}$-adversaries must involve $\Omega(Ln^{1+1/d})$ bits of expected communication, if the protocol must achieve output safety with at least $0.5 + \varepsilon$ probability for any absolute constant $\varepsilon > 0$.
}

\apxthm{theorem}{abathm}{
    For all $d \geq 2$ and $L = \Omega(n^{1-1/d})$, a terminating asynchronous byzantine agreement protocol for $L$-bit inputs that can tolerate all efficient send-omission $\Z_\pj^{n,d}$-adversaries must involve $\Omega(Ln^{1+1/d})$ bits of expected communication, if the protocol must achieve output safety with at least $0.75 + \varepsilon$ probability for any absolute constant $\varepsilon > 0$.  
}

\apxthm{theorem}{acsthm}{
    For all $d \geq 2$ and all $L \geq 1$, a terminating asynchronous core set agreement protocol for $L$-bit inputs that can tolerate all efficient send-omission $\Z_\tpj^{n,d}$-adversaries must involve $\Omega(Ln^{2+1/d})$ bits of expected communication, if the protocol must achieve output safety with at least $0.875 + \varepsilon$ probability for any absolute constant $\varepsilon > 0$.  
}

\subparagraph{The Importance of Termination.} Our lower bounds for the asynchronous setting require termination. To show that this requirement is essential, in Section \ref{pullcastsection} we present an error-free reliable broadcast protocol for $L$-bit inputs which can for any parameter $\delta > 1$ be tuned to cost at most $(1 + \frac{1}{\delta-1})Ln + O(\delta n^2\log(\delta n))$ bits of communication and tolerates any general-omission \linebreak adversary. The caveat is that the parties cannot terminate it after they output, as the protocol uses a pull-based approach where the parties who know the sender's input $v$ must stay alive to serve other parties' dynamic requests for erasure-coded fragments of $v$.\footnote{It is the impossibility of such a dynamic approach that makes our lower bounds hold against terminating protocols. In a terminating protocol, we can partition the parties into two sets $U, U'$ such that $U' \in \Z$, and force the parties in $U$ to terminate without them ever receiving any messages from $U'$. The parties in $U$ will then be forced to blindly send their common output to the parties in $U'$, with the redundancy required for everyone in $U'$ to terminate.  If on the other hand the protocol were non-terminating, then the parties in $U$ could wait to hear from the parties in $U'$ and reply to them dynamically.} We then discuss how this approach can be extended for byzantine fault tolerance and why our other lower bounds for asynchronous byzantine agreement and core set agreement also require termination. In Section \ref{discuss}, we also give a simple crash fault tolerant terminating reliable broadcast protocol \linebreak that costs $L(n-1) + O(n^2)$ bits to show that our lower bounds for asynchronous send-omission tolerant protocols do not extend to crash fault tolerant protocols.\footnote{In the crash fault setting we assume no message clawbacks, which means that if a party sends a message, then the adversary must eventually deliver the message even if the party later crashes.} A consequence of our two reliable broadcast protocols is that Locher's $1.5(1 - o(1))Ln$ bit communication complexity lower bound for a certain class of reliable broadcast protocols that can tolerate $t < \frac{n}{3}$ crash faults \cite{locher24} does not extend to all non-terminating reliable broadcast protocols: The bound does not apply to all non-terminating protocols in the omission/byzantine fault settings, and it does not apply to all terminating protocols in the crash fault setting.

\subparagraph{The Requirements on $\bm L$.} Most of our lower bounds require minimum input lengths. This is because the parties can (ab)use silence to encode information into metadata. For example, in a synchronous network, the party set $[n - 1]$ can transmit a tuple $(r, j) \in [R] \times [n-1]$ to the party $n$ as follows: On round $r \in [R]$, the party $j \in [n-1]$ sends a single bit to $i$. This way, $i$ learns $O(\log R + \log n)$ bits of information while receiving only one bit, breaking the assumption we would like to make that to learn $L$ bits a party must receive $\Omega(L)$ bits in expectation. To make this assumption true, we lower bound $L$, and thus limit the effectiveness of silence. We note that this would not be required if we were to measure communication with entropy instead of bits sent. Less abstractly, we would not need lower bounds on $L$ if we assumed that each message is tagged with the pertinent metadata (a sender ID and a receiver ID, plus a round number for synchronous protocols). Adopting this tagged model would also make our lower bounds for synchronous protocols hold no matter the protocol's round complexity, even for Las Vegas style protocols with variable round complexities.

\subparagraph{Tightness.} When the adversary structure satisfies the $\Q^d$ condition, our lower bounds are of the form $\Omega(L_\mathsf{out} \cdot n^{1+1/d})$, where $L_\mathsf{out}$ is the output length. While we can only conjecture that this bound is tight in the synchronous setting, we know that it is tight in the asynchronous setting for all large $L$. We show this in Section \ref{termsection} via an error-free general-omission tolerant termination protocol (inspired by the protocols in \cite{nayak20}) that lets any asynchronous agreement protocol terminate with $O(Ln^{1+1/d} + n^2\log n)$ bits of additional communication, for any $\Q^d$-satisfying adversary structure $\Z$. To obtain this result, we find some non-empty party set $G \subseteq [n]$ such that $|G \setminus Z| \geq |G|n^{-1/d}$ for all $Z \in \Z$, which always exists by Proposition \ref{denseprop}. Then, we task each party in $G$ (which contains at least $|G|n^{-1/d}$ correct parties) with sending \linebreak everyone an erasure code symbol of size $O(\frac{Ln^{1/d}}{|G|} + \log n)$ such that any $|G|n^{-1/d}$ of the $|G|$ symbols suffice for everyone to learn the output and terminate.

\subparagraph{Other Agreement Tasks.} While we focus on five fundamental agreement tasks in the paper, our lower bounds generalize to other tasks, including ones with weaker agreement guarantees. In particular, our lower bounds for byzantine agreement also hold for a weaker variant of the task called adopt-commit \cite{gafni98}.\footnote{Adopt-commit is sometimes called graded consensus, but the latter can also refer to a slightly stronger primitive. We refer the reader to \cite{aw24} for modern definitions and a discussion of the nuances.} While for brevity we omit a formal proof of this, the intuitive reason is that an adopt-commit protocol must (like a byzantine agreement protocol) be such that if a quorum $S$ of correct parties (whose complement the adversary could corrupt) have a common input $v$, then every correct party outside $S$ learns $v$. This induces the same entropy lower bounds for adopt-commit and byzantine agreement. Alternatively, an easier argument that works for large $L$ is that one can reduce multivalued byzantine agreement to an instance of multivalued adopt-commit followed by an instance of binary byzantine agreement (with a cost not dependent on $L$), as Turpin and Coan show in \cite{tc84} with older terminology. For similar \linebreak reasons, our bound for byzantine broadcast also holds for its weaker variant gradecast \cite{fm88}.

\subsection{Technical Overview} \label{tech-over}

Below, we sketch the proof of our simplest lower bound: the $\Omega(Ln^{1+1/d})$ bit one for reliable broadcast (Theorem \ref{rbthm}). Afterwards, we will explain how our other lower bound proofs differ. To keep the sketch simple, let us only consider deterministic and error-free reliable broadcast.

For an adversary structure $\Z$, let the corresponding quorum structure be $\S = \{S \subseteq [n]:[n] \setminus S \in \Z\}$. When we run a protocol against a $\Z$-adversary, the guarantee we have is that the set of correct parties will always be in $\S$. Note that the quorum structure $\S$ is up-closed (with the maximum element $[n]$) since $\Z$ is down-closed, and the $\Q^d$ condition is equivalent to \linebreak the statement that the intersection of any $d$ sets in $\S$ must be non-empty.

We will define the adversary structure $\Z_\pj^{n,d}$ such that the corresponding quorum structure $\S_\pj^{n,d}$ has a subset $\S^*$ with the following properties: \begin{itemize}
    \item $\bigcup \S^*$ (the union of all sets in $\S^*$) and $[n] \setminus \bigcup \S^*$ both contain $\Theta(n)$ parties.
    \item Every party in $\bigcup \S^*$ belongs to $\Theta(n^{-1/d}|\S^*|)$ of the sets in $\S^*$.
\end{itemize}

Note that the properties above imply that the quorum sets in $\S^*$ are of size $O(n^{-1/d}|\S^*|)$ on average. For our lower bounds, we will get $\S^*$ via $d$-dimensional finite projective geometry \cite{beutel}, and each quorum in $\S^*$ will be a hyperplane of points, with each point representing a party. To get such a quorum structure, we must let the adversary corrupt $n(1 - o(1))$ parties. This is why our lower bounds require general adversaries.

Let $U = \bigcup \S^*$. Consider the execution of a terminating reliable broadcast protocol where the sender $s$ is in $U' = [n] \setminus U$ and has some uniformly random $L$-bit input $v$. Every party is correct, including the sender. Let us say that this execution happens in the world $W_v$.

In the world $W_v$, the adversary delays every message sent by the parties outside $U \cup \{s\}$. This way, it forces the parties in $U$ to terminate (with the output $v$ since this is the sender's input) without them ever receiving any message from the parties outside $U \cup \{s\}$. The parties in $U$ have to terminate without waiting to hear from the parties outside $U \cup \{s\}$ since as far as they know, the parties outside $U \cup \{s\}$ could all be silent faulty parties.

We argue that for every $S \in \S^*$ and every $i \in U' \setminus \{s\}$, the parties in $S$ must in expectation send $\Omega(L)$ bits to $i$ before terminating. This is because as far as the parties in $S$ know, it is possible that every party outside $S \cup \{i\}$ is faulty, and if so, then $i$ has to learn the uniformly random $L$-bit output $v$ of the parties in $S$ solely based on that messages that it receives from $S$. By an entropy argument, $i$ must receive $\Omega(L)$ bits in expectation from $S$ for this.

We have argued above that $\E[\sum_{j \in S}B_{j \rightarrow i}] = \Omega(L)$ for all $S \in \S^*$ and $i \in U' \setminus \{s\}$, where $B_{j \rightarrow i}$ is the number of bits that $j$ sends to $i$ in a random fault-free execution. Consequently, $\E[\sum_{S \in \S^*}\sum_{j \in S}\sum_{i \in U' \setminus \{s\}}B_{j \rightarrow i}] = \Omega(|\S^*|(|U'| - 1)L) = \Omega(|\S^*|Ln)$. Finally, we observe that a bit that a party in $U$ sends can only contribute $O(n^{-1/d}|\S^*|)$ to this $\Omega(|\S^*|Ln)$ expectation since every party in $U$ is in $\Theta(n^{-1/d}|\S^*|)$ of the sets $S \in \S^*$, which means that the protocol must involve $\Omega(\frac{|\S^*|Ln}{n^{-1/d}|\S^*|}) = \Omega(Ln^{1+1/d})$ bits of communication in expectation.

Our $\Omega(Ln^{1+1/d})$ lower bound for asynchronous byzantine agreement will involve a similar proof, with some more subtleties regarding the protocol's safety probability. The other lower \linebreak bounds (for core set agreement and interactive consistency) will be a bit harder to get. \begin{itemize}
    \item \textit{Core Set Agreement.} For core set agreement, a similar proof is sufficient to show that one needs $\Omega(L|C|n^{1+1/d})$ bits of expected communication, where $C$ is the agreed upon core set. To get the lower bound $\Omega(Ln^{2+1/d})$, we must also ensure that $|C| = \Omega(n)$. We will define \linebreak $\S_\pj^{n,d}$ such that all of its minimal quorums (core set candidates) are of size $\Theta(n^{1-d})$, which means that $\S_\pj^{n,d}$ will not suffice for the core set agreement lower bound. Hence, we will need \linebreak a more advanced quorum structure $\S_\tpj^{n,d}$ which still contains many minimal quorums of size $\Theta(n^{1-d})$, but also some of size $\Theta(n)$. In our proof, the adversary will initially only allow a minimal party quorum $C$ of size $\Theta(n)$ to communicate, and thus force $C$ to be the core set. The remaining steps will be similar to what we sketched for reliable broadcast.
    \item \textit{Interactive Consistency.} As we will study interactive consistency in synchronous networks, we cannot rely on termination. Instead, we will rely on the fact that if the parties do not communicate $\Omega(Ln^{2 + 1/d})$ bits in expectation, then one can find two worlds $W_V$ and $W_{V'}$ (both with no faults), a set $S \in \S^*$ and a party $i \in [n] \setminus \bigcup \S^*$ such that the bidirectional communication between $S$ and $i$ is the same in both worlds even though some party $j$ has a different input in the two worlds. Hence, by corrupting every party outside $S \cup \{i\}$, the adversary can fool the parties in $S$ into thinking they are in $W_V$ while fooling $i$ into thinking it is in $W_{V'}$, causing $S$ and $i$ to disagree on $j$'s input. Since this strategy requires the faulty parties to perfectly simulate how they would have behaved as correct parties in different worlds, the faulty parties must be byzantine and the protocol must be error-free.
\end{itemize}

\section{Related Work}

The literature on consensus is mostly centered around $t$-threshold adversaries that can corrupt $t$ parties. In this setting, the communication complexity lower bounds for agreement on long inputs are usually simple, and the challenge lies in matching them with as little overhead as possible. This challenge began with \cite{fh06}, which proved the lower bounds $\Omega(Ln)$ for byzantine broadcast and $\Omega(Lt)$ for byzantine agreement, and solved these tasks with $O(n^3\kappa)$ bits against $t < \frac{n}{3}$ faults in a synchronous network (where $\kappa$ is a security parameter). Later works have designed deterministic error-free protocols that cost $O(Ln + n^2\log n)$ bits for synchronous byzantine agreement and broadcast \cite{chen21}, for reliable broadcast \cite{long22}, and for extending binary asynchronous byzantine agreement protocols into multivalued ones \cite{aa25}, all with the optimal byzantine fault tolerance $t < \frac{n}{3}$. Protocols that support long inputs efficiently also exist for core set agreement \cite{dumbo}, as well as for authenticated byzantine agreement and interactive consistency against $t < \frac{n}{2}$ faults \cite{nayak20}. It is also known that there exist authenticated byzantine broadcast protocols for $t < n$ faults that cost $O(Ln + \kappa \cdot \mathit{poly}(n))$ bits of communication \cite{co18}.

The threshold adversary setting is well-known for its quadratic message complexity lower bounds, such as the ones in \cite{dr85,a19}. Most famously, Dolev and Reischuk prove in their seminal \linebreak work \cite{dr85} that a synchronous deterministic byzantine broadcast protocol which can tolerate $t$ byzantine faults must involve $\Omega(n + t^2)$ messages. Informally, the Dolev-Reischuk argument \cite{dr85} for deterministic protocols is that if a party $i$ normally receives $o(t)$ messages in fault-free executions, then, in an execution with faults, the adversary could corrupt every party who would normally send $i$ a message, and thus force $i$ to output without it receiving any messages. Curiously, the adversary structures we use in this work are not conducive to this argument, as they do not allow the adversary to corrupt any $t$ parties it chooses. We suspect that these adversary structures thus admit byzantine broadcast protocols that cost $o(n^2)$ messages/bits, though we leave it for future work to show this.

The lower bound most like ours is the one in \cite{locher24} by Locher, which was influential for our work. In \cite{locher24}, Locher proves a communication complexity lower bound of $(\frac{Ln^2}{n-t})(1 - o(1))$ bits (or $(1.5 - o(1))Ln$ bits when $n = 3t + 1$) for $3$-round crash fault tolerant reliable broadcast protocols where the sender sends $o(Ln)$ bits in the first round. The bound relies on a double counting argument comparable to the one we use, and it is tight for the considered class of protocols \cite{ls25,ls26}. Locher conjectured in \cite{locher24} that his bound might extend to reliable broadcast \linebreak protocols outside the class of protocols he considered. We show that it does not extend to all reliable broadcast protocols: The omission tolerant non-terminating reliable broadcast protocol in Section \ref{pullcastsection} is not subject to it, and neither is a crash fault tolerant terminating reliable broadcast protocol that we describe in Section \ref{discuss}.

The fact that termination is the bottleneck in our asynchronous lower bounds rather than the rest of the protocol is noteworthy, given that termination is usually an easy-to-achieve goal handled with a simple protocol (e.g.\ by reliable agreement in \cite{ddlmrs24}). And yet, with respect to the lower bound, it turns out that termination is the main challenge. This reminds us of a recent work on reliable broadcast \cite{mw24}, where the authors define a reliable broadcast variant with very strong termination guarantees that can to some extent tolerate the correct parties prematurely quitting the protocol, but then they fail to give a protocol which satisfies these guarantees with $O(Ln + \dots)$ bits of communication. We are also reminded of the recent work \cite{lewis2026}, which shows for synchronous binary byzantine agreement against $t < n(1/3 - \varepsilon)$ byzantine faults that a deterministic protocol for this task can make all but $\varepsilon n$ of the correct parties learn the output (and thus achieve univalency) with $O(n\log n)$ bits of communication, and the part that costs $\Theta(n^2)$ bits is making the remaining parties learn the output.

Without faults, $n$ parties can compute any function of their inputs with $O(Ln^2)$ bits of communication, since this is the complexity of a single all-to-all input exchange. With faults, $O(Ln^2)$ bits can be insufficient, as our work shows. To our knowledge, this was not known before for any task that does not require the parties to keep their inputs secret. 

A secrecy-involving task which can require a massive amount of communication is secret sharing, where each party acquires a share of a secret such that the parties can reconstruct the secret iff a sufficient quorum of parties combine their shares. It has been conjectured that general adversaries can require some parties to learn shares which are $2^{\Theta(n)}$ times larger than the secret \cite{beimel}.\footnote{This conjecture has been proven for certain kinds of secret sharing schemes \cite{linearsss,ls20}.} Accordingly, general adversary tolerant protocols based on secret sharing such as the $\Q^3$ adversary tolerant asynchronous byzantine agreement protocol in \cite{c23} can have exponential-in-$n$ communication complexities. Note that secret sharing has efficient solutions against $t$-threshold adversaries, where any $t + 1$ shares suffice for secret reconstruction \cite{shamir}.

Another secrecy-involving task related to consensus is secure multi-party computation (MPC), where the parties must compute a function of their inputs while keeping their inputs secret. MPC's input secrecy has been used to prove communication complexity lower bounds for it; see for example \cite{dls21}. In this work, there is no input secrecy. Instead, we only use the idea that a party has to receive $\Omega(L_\mathsf{out})$ bits from correct parties to learn an $L_\mathsf{out}$-bit output.

The adversary structures $\Z_\pj^{n,d}$ and $\Z_\tpj^{n,d}$ we use in this paper are based on finite projective geometry. In particular, for $\Z_\pj^{n,d}$, the family of quorum sets $\S^*$ which we used in Section \ref{tech-over} is such that each party in $\bigcup \S^*$ is identified by a unique point in the $d$-dimensional projective geometry $PG(d,q)$ of order $q$ \cite{beutel}, where $q = \Theta(n^{1/d})$ is a prime power, and each quorum set in $\S^*$ is a hyperplane of parties in $PG(d,q)$. This quorum structure is inspired by Maekawa \cite{maekawa}, who used the lines of $PG(2,q)$ as quorums for distributed mutual exclusion. For all the facts about $PG(d,q)$ we will use in this work, we refer the reader to the first chapter of \cite{beutel}.

\section{Model}

We assume a network of $n$ parties with the IDs $\{1,\dots,n\} = [n]$ which are pairwise connected to each other via reliable and authenticated channels. An adversary characterized by an adversary structure $\Z$ (a down-closed family of subsets of $[n]$) corrupts a set of parties $Z \in \Z$, where $\Z$ satisfies the $\Q^d$ condition for any $d \geq 1$ if no $d$ or fewer sets in $\Z$ have the union $[n]$. A corrupted party may be byzantine, and thus behave arbitrarily, or it may be omission faulty. An omission faulty party conforms to the prescribed protocol; however, the adversary can drop any message that the party sends (if the party is send-omission faulty), as well as any message sent to the party (if the party is general-omission faulty). 

Our lower bounds hold for static adversaries, which must choose the set $Z \in \Z$ of corrupt parties before the parties start running the protocol. However, our protocols in Section \ref{pullcastsection} and Section \ref{termsection} can tolerate adaptive adversaries which can corrupt parties during protocol execution, depending on the information they gather while the protocol runs.

The network may be synchronous or asynchronous. In a synchronous network, a protocol proceeds in synchronous rounds $1,2,\dots$. Each round consists of three steps where each party \linebreak sends messages to the other parties, receives the messages from the other parties, and updates its state based on the messages it received. An $R$-round synchronous protocol consists of $R$ rounds $1,2,\dots,R$, and it ends with each party outputting a value at the end of round $R$.

In asynchronous networks, on the other hand, there are no rounds, and a party only sends \linebreak messages in response to beginning a protocol or receiving a message. The adversary controls the message scheduling, which means that it can arbitrarily delay or reorder all sent messages; however, it is required to eventually deliver every message whose sender is correct.

Our lower bounds apply to randomized protocols.\footnote{We note that our lower bounds apply to deterministic protocols as well, since deterministic protocols are a special case of randomized ones.} We assume that there exists a vector of $n$ randomness tapes $\R = (\R_1, \dots, \R_n)$, and each party $i$ has private access to $\R_i$. Each tape $\R_i \in \{0,1\}^\omega$ is an infinite string of bits. In the synchronous setting, we assume that the tapes are mutually independent strings of i.i.d.\ uniformly random bits, while we allow the strings to be non-uniform and correlated in the asynchronous model. A correct party $i$ acts deterministically based on its local view (on its input if it has any, on $\R_i$ and on the messages it has received so far). The tape vector $\R$ is independent of the parties' inputs.

To keep our model simple, we omit explicit setup parameters such as a PKI (public key infrastructure) or a CRS (common random string) from it. However, our lower bounds also hold for protocols that use finite setups. For error-free synchronous protocols, this is because the protocol must stay error-free for any fixing of the joint setup, and conditioning on a fixed setup reduces the analysis to the setup-free case. As for asynchronous protocols, the parties' possibly non-uniform and correlated randomness tapes already subsume any joint setup.

Our lower bounds are for the expected communication complexity of a protocol's fault-free executions. We express this expected complexity by $\E[\sum_{i \in [n]}\sum_{j \in [n] \setminus \{i\}}B_{i\rightarrow j}]$, where $B_{i \rightarrow j}$ represents the total number of bits that the party $i$  sends to the party $j$ in a random execution. This expectation is over the input assignment distribution $\V$ (which assigns a private input to each party who should have one), plus the joint randomness tapes $\R$ of the parties.\footnote{In the terminology of communication complexity, our communication complexity model is known as the number-in-hand (NIH) message-passing model \cite{pvz12}.}

We say that an agreement protocol achieves output safety with probability $p$ if it achieves validity and agreement with at least this probability. Formally, for any input distribution, the probability that either any party obtains an output which violates the protocol's validity guarantees or any two correct parties obtain distinct outputs must be at most $1 - p$.

In our lower bounds, we use the variable $T_{S \rightarrow i}$ to refer to the transcript of communication between a set of parties $S$ to a party $i$. The arrow denotes the direction of communication: $T_{S \rightarrow i}$ only records the communication that $i$ receives from $S$ (with all the pertinent metadata which $i$ may use in the protocol), while $T_{S \leftrightarrow i}$ also records the communication that the parties in $S$ receive from $i$ (again with the pertinent metadata). The idea is that if the transcript $T$ between $S$ and $i$ has a sufficiently large Shannon entropy $H(T)$, then we can transform the transcript entropy lower bound into a communication complexity one by doing the following: \begin{enumerate}
    \item encode $T$ into an array $A$, such that the array's cells are filled with the raw messages the parties send each other while the cell indices implicitly carry the metadata,
    \item argue that the cross communication between $S$ and $i$ must be $\Omega(H(T))$ bits in expectation because the array's cells must in total contain $\Omega(H(T))$ bits in expectation. 
\end{enumerate}
With the transcript encoding schemes we choose, $H(T)$ is sufficiently large if $H(T) = \Omega(R|S|)$ (for an $R$-round synchronous protocol), or if $H(T) = \Omega(|S|)$ (for an asynchronous protocol).\footnote{If the protocol is asynchronous, then we make assumptions regarding the adversary's message scheduling (which we respect in our lower bound proofs) to reduce the encoding overhead.} These lower bounds on $H(T)$ ensure that the information content of the bits sent dominate the information content of metadata. We discuss the technicalities in the \hyperref[apxsec]{appendix}.

\section{Adversary Structures} \label{advstructures}

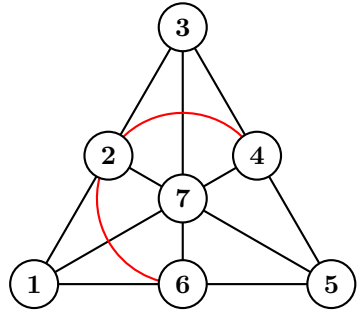
\begin{figure}[ht]
    \centering
    \begin{tikzpicture}[scale=0.98]
        \coordinate (N1) at (0, 0);
        \coordinate (N3) at (2, {2*sqrt(3)});
        \coordinate (N5) at (4, 0);
        \coordinate (N2) at (1, {sqrt(3)});
        \coordinate (N4) at (3, {sqrt(3)});
        \coordinate (N6) at (2, 0);
        \coordinate (N7) at (2, {sqrt(3)/1.5});
        \draw[thick] (N1) -- (N4);
        \draw[thick] (N5) -- (N2);
        \draw[thick] (N3) -- (N6);
        \draw[thick] (N1) -- (N3) -- (N5) -- cycle;
        \draw[thick, red] (N4) arc[start angle=30, end angle=270, radius={sqrt(3)/1.5}];
        \tikzset{fano node/.style={
            circle, draw=black, fill=white, thick, inner sep=0pt, minimum size=18pt, font=\bfseries
        }}
        \foreach \i in {1,...,7}
            \node[fano node] at (N\i) {\i};
    \end{tikzpicture}
    \caption{Fano's Plane, a depiction of $PG(2,2)$. The hyperplanes (lines) of $PG(2,2)$ are the circular arc $\{4,6,2\}$ and the $6$ collinearly drawn triples $\{1,2,3\},\{3,4,5\},\{5,6,1\},\{2,7,5\},\{4,7,1\},\{6,7,3\}$. In $\Z_\pj^{14,2}$, the adversary can choose any of these $7$ triples, and corrupt every party in $[n] = \{1,\dots,14\}$ excluding the parties in the triple. We note that the labels of the parties/nodes are arbitrary.}
    \label{fig:placeholder}
\end{figure}

Now, we define the adversary structure families $\Z_\pj^{n,d}$ and $\Z_\tpj^{n,d}$. The simpler structure $\Z_\pj^{n,d}$ exists for any $d \geq 1$ if $n = 2\sum_{k = 0}^dq^k = 2(q^{d+1} - 1)/(q - 1)$ for a prime power $q = \Theta(n^{1/d})$. To define it, we partition the parties into the sets $U = [\frac{n}{2}]$ and $U' = [n] \setminus U$. We bijectively label each party $i \in U$ with a point $p(i)$ in the $d$-dimensional finite projective geometry $PG(d,q)$ of order $q$. Finally, we define the quorum structure $\S_\pj^{n,d} = \{S \subseteq [n] : [n] \setminus S \in \Z_\pj^{n,d}\}$ that corresponds to $\Z_\pj^{n,d}$ so that for any party set $S \subseteq [n]$, $S \in \S_\pj^{n,d}$ iff there exists a hyperplane $\H$ of $PG(d,q)$ such that $S$ is a superset of the minimal quorum $S_\H = \{i \in U : p(i) \in \H\} \in \S_\pj^{n,d}$. In other words, we allow the adversary to do the following:
\begin{enumerate}
    \item pick any hyperplane $\H$ out of the $(q^{d+1} - 1)/(q - 1)$ hyperplanes of $PG(d,q)$,
    \item corrupt any subset of the parties in \{$i \in U:p(i) \not \in \H\} \cup U'$.
\end{enumerate}

This adversary structure satisfies the $\Q^d$ condition since the intersection of $d$ hyperplanes is never empty. Moreover, for each party $i \in U$, the proportion of quorums in $\S_\pj^{n,d}$ which contain $i$ is exactly $(q^d - 1)/(q^{d + 1} - 1) = \Theta(n^{-1/d})$, by the incidence structure of $PG(d,q)$.

The adversary structure $\Z_\tpj^{n,d}$ is somewhat more complicated. It exists for any $d \geq 1$ if $n = 3(q^{d+1} - 1)/(q - 1)$ for a prime power $q$. This time, we partition the party set $[n]$ into three sets: $U_1 = \{1,\dots,\frac{n}{3}\}, U_2 = \{\frac{n}{3}+1,\dots,\frac{2n}{3}\}$ and $U' = \{\frac{2n}{3}+1,\dots,n\}$. For every $i \in [\frac{n}{3}]$, we label the party $i \in U_1$ and the party $i + \frac{n}{3} \in U_2$ with some unique point $p(i) = p(i + \frac{n}{3})$ in $PG(d,q)$. The adversary can once again corrupt any of the parties in $U'$. Plus, it can do the following twice: pick a hyperplane, and corrupt every party in one of $U_1$ or $U_2$ whose label is not incident to this hyperplane. In other words, the adversary can pick any two (possibly identical) hyperplanes $\H,\H'$ of $PG(d,q)$ and do one of the following:
\begin{itemize}
    \item corrupt any subset of the parties in $\{i \in U_1:p(i) \not \in \H \cap \H'\} \cup U'$,
    \item corrupt any subset of the parties in $\{i \in U_1:p(i) \not \in \H\} \cup \{i \in U_2:p(i) \not \in \H'\} \cup U'$,
    \item corrupt any subset of the parties in $\{i \in U_2:p(i) \not \in \H \cap \H'\} \cup U'$.
\end{itemize}

Let us say that the adversary hits $U_1$ or $U_2$ if it picks a hyperplane $\H$ and corrupts some of the parties in $U_1$ or $U_2$ whose labels are not incident on $\H$. The adversary can perform two hits,  either twice on one of $U_1$ or $U_2$ or once on each. The adversary structure $\Z_\tpj^{n,d}$ satisfies the $\Q^d$ condition because taking the union of $d$ maximal adversary sets means performing $2d$ hits, which is not enough to corrupt every party. The reason for this is that the hyperplanes of \linebreak $PG(d,q)$ form a $d$-wise intersecting family, which means that the adversary must hit both $U_1$ and $U_2$ at least $d + 1$ times (i.e.\ perform $2d + 2$ hits) to fully corrupt both $U_1$ and $U_2$.

The quorum structure $\S_\tpj^{n,d}$ corresponding to $\Z_\tpj^{n,d}$ has a substructure $\S^*$ which contains one quorum $S_\H$ for each hyperplane $\H$ of $PG(d,q)$, where $S_\H = \{i \in U_1 \cup U_2:p(i) \in \H\}$. By the incidence structure of $PG(d,q)$, the proportion of quorums in $\S^*$ that contain each party $i \in U_1 \cup U_2$ is $(q^d - 1)/(q^{d+1} - 1) = \Theta(n^{-1/d})$.

The structure $\S_\tpj^{n,d}$ also has a minimal quorum $C$ of size at least $(q^{d+1} - 1)/(q-1) \geq \frac{n}{3}$, which can be obtained by choosing two distinct hyperplanes $\H$ and $\H'$ of $PG(d,q)$ and letting $C = U_1 \cup \{i \in U_2:p(i) \in \H \cap \H'\}$. If $d = 1$, then $C = U_1$, and otherwise, $C$ is a collection of $\frac{q^{d+1} - 1}{q-1} + \frac{q^{d-1} - 1}{q - 1}$ parties, where the latter term $\frac{q^{d-1} - 1}{q - 1} = \sum_{k=0}^{d-2}q^k$ is the number of points that any $(d-2)$-dimensional subspace of $PG(d,q)$ contains.

\section{Synchronous Lower Bounds} \label{syncsection}

We begin with a combinatorial lemma which relates the expected amount of communication that takes place between many (quorum, party) pairs in a protocol with the protocol's total expected communication complexity. We will use this lemma throughout the paper.

\begin{lemma} \label{comb-lemma}
    Let $\S$ be a family of party sets such that each party in $U = \bigcup \S$ is in at most $Q = O(n^{-1/d}|\S|)$ of the sets in $\S$, for any $d \geq 1$. Suppose there exists some $E > 0$ such that a protocol's random fault-free executions satisfy $\E[\sum_{j \in S}(B_{j \rightarrow i} + B_{i \rightarrow j})] \geq E$ for all $S \in \S$ and all parties $i$ in a set $U' \subseteq [n] \setminus U$ of size $\Theta(n)$. Then, the random fault-free executions of the protocol involve $\Omega(En^{1+1/d})$ bits of total expected communication.
\end{lemma} 

\begin{proof}
    The proof is by double counting. The expected fault-free communication complexity of the protocol is $\E[\sum_{j \in [n]}\sum_{i \in [n] \setminus \{j\}}B_{j \rightarrow i}]$, which is at least $\E[\sum_{j \in U}\sum_{i \in U'}(B_{j \rightarrow i} + B_{i \rightarrow j})]$.
    
    Due to the linearity of expectation, summing $\E[\sum_{j \in S}(B_{j \rightarrow i} + B_{i \rightarrow j})] \geq E$ up over every $S \in \S$ and $i \in U'$ results in $\E[\sum_{S \in \S}\sum_{j \in S}\sum_{i \in U'}(B_{j \rightarrow i} + B_{i \rightarrow j})] \geq E|\S||U'| = \Omega(En|\S|)$. We have $\sum_{S \in \S}\sum_{j \in S}\sum_{i \in U'}(B_{j \rightarrow i} + B_{i \rightarrow j}) \leq Q(\sum_{j \in U}\sum_{i \in U'}(B_{j \rightarrow i} + B_{i \rightarrow j}))$ since for all $j \in U$ and $i \in U'$, the term $(B_{j \rightarrow i} + B_{i \rightarrow j})$ appears in the first sum at most $Q$ times (once for each $S \in \S$ that contains $j$). Therefore, the expected communication complexity of the protocol is lower bounded by $\E[\sum_{j \in U}\sum_{i \in U'}(B_{j \rightarrow i} + B_{i \rightarrow j})] \geq \frac{1}{Q} \cdot \E[\sum_{S \in \S}\sum_{j \in S}\sum_{i \in U'}(B_{j \rightarrow i} + B_{i \rightarrow j})] \linebreak = \Omega(\frac{En|\S|}{Q}) = \Omega(\frac{En|\S|}{n^{-1/d}|\S|}) = \Omega(En^{1+1/d})$.
\end{proof}

Now, we are ready to prove the lower bound for interactive consistency.

\interthm

\begin{proof}
    Recall that $\Z_\pj^{n,d}$ partitions the parties into two sets $U, U'$ of size $\frac{n}{2}$, where a quorum is a hyperplane of $U$. We consider the input distribution $\V$ which assigns every party an i.i.d.\ uniformly random $L$-bit input. This distribution has $Ln$ bits of entropy.

    Fix any minimal quorum $S_\H$ of $\S_\pj^{n,d}$ (obtained by picking the parties in $U$ whose labels are incident to the hyperplane $\H$ of $PG(d,q)$), and fix a party $i \in U'$. Consider the random transcript $T_{S_\H \leftrightarrow i}$ of the bidirectional communication between $S_\H$ and $i$ in a fault-free execution, which is randomized by the input distribution $\V$ and party randomnesses $\R$. We claim that $H(T_{S_\H \leftrightarrow i}) \geq H(\V) = Ln$. Assume this is correct for now. By the assumption $L = \Omega(Rn^{-1/d})$, we have $H(T_{S_\H \leftrightarrow i}) = \Omega(Rn^{1-1/d})$, which means by Lemma \ref{sync-cov-lemma} (proven in the appendix) that $\E_{\V,\R}[\sum_{j \in S_\H}(B_{j \rightarrow i} + B_{i \rightarrow j})] = \Omega(H(T_{S_\H \leftrightarrow i})) = \Omega(Ln)$. Hence, by Lemma \ref{comb-lemma}, the protocol involves $\Omega(Ln^{2+1/d})$ bits of expected communication.    

    What remains to prove is that the claim $H(T_{S_\H \leftrightarrow i}) \geq H(\V)$ is true. Suppose not. Then, since $H(\V) > H(T_{S_\H \leftrightarrow i})$, we have $H(V \mid T_{S_\H \leftrightarrow i}) \geq H(\V) - H(T_{S_\H \leftrightarrow i}) > 0$. The conditional entropy $H(V \mid T_{S_\H \leftrightarrow i})$ equals $\sum_{\tau_{S_\H \leftrightarrow i}} \Pr[T_{S_\H \leftrightarrow i} = \tau_{S_\H \leftrightarrow i}] \cdot H(\V \mid T_{S_\H \leftrightarrow i} = \tau_{S_\H \leftrightarrow i})$, where $\tau_{S_\H \leftrightarrow i}$ refers to any fixed finite $R$-round transcript which $T_{S_\H \leftrightarrow i}$ might equal, and this sum is positive. So, there must be some $\tau_{S_\H \leftrightarrow i}$ such that $\Pr[T_{S_\H \leftrightarrow i} = \tau_{S_\H \leftrightarrow i}] \cdot H(\V \mid T_{S_\H \leftrightarrow i} = \tau_{S_\H \leftrightarrow i})$ is positive. In particular, the factor $H(\V \mid T_{S_\H \leftrightarrow i} = \tau_{S_\H \leftrightarrow i})$ must be positive. For this, there must exist two distinct $V, V'$ in the support of $\V$ such that $\Pr[\V = V \mid T_{S_\H \leftrightarrow i} = \tau_{S_\H \leftrightarrow i}] > 0$ and $\Pr[\V = V' \mid T_{S_\H \leftrightarrow i} = \tau_{S_\H \leftrightarrow i}] > 0$. In the rest of our proof below, we use these two input assignments $V,V'$ and the fixed transcript $\tau_{S_\H \leftrightarrow i}$ to show that the protocol is not error-free. This contradiction means that the assumption $H(T_{S_\H \leftrightarrow i}) < H(\V)$ must be false.

    Observe that $\Pr[\V = V \mid T_{S_\H \leftrightarrow i} = \tau_{S_\H \leftrightarrow i}] > 0$ implies $\Pr[T_{S_\H \leftrightarrow i} = \tau_{S_\H \leftrightarrow i} \mid \V = V] = p$ for some $p > 0$ by the definition of conditional probability and by Bayes' Theorem. Similarly, $\Pr[\V = V' \mid T_{S_\H \leftrightarrow i} = \tau_{S_\H \leftrightarrow i}] > 0$ implies $\Pr[T_{S_\H \leftrightarrow i} = \tau_{S_\H \leftrightarrow i} \mid \V = V'] = q$ for some $q > 0$.

    Let $\W_V$ (resp.\ $\W_{V'}$) be the set of all worlds where a fault-free execution with the input assignment $V$ (resp.\ $V'$) takes place and the transcript between $S_\H$ and $i$ is $\tau_{S_\H \leftrightarrow i}$. Interactive consistency requires the parties to all output $V$ in the worlds in $\W_V$ and $V'$ in the worlds in $\W_{V'}$. Let us define a hybrid world set $\W_\mathsf{Hyb}$ where the parties in $S_\H \cup \{i\}$ are correct while the rest of the parties are byzantine, $i$'s input is assigned in accordance with $V$, and the parties in $S_\H$ have their inputs assigned in accordance with $V'$. Every fixed world $W_{\mathsf{Hyb}} \in \W_\mathsf{Hyb}$ is such that $i$'s randomness $\R_i$ in $W_\mathsf{Hyb}$ is the same as it is in some fixed world $W_V \in \W_V$, and similarly the randomnesses $\{\R_s:s \in S_\H\}$ of the parties in $S_\H$ are the same as they are in some fixed world $W_{V'} \in \W_{V'}$. Below, we describe a cut-and-paste attack the adversary can perform in the world $W_\mathsf{Hyb}$ to cause $i$ to output $V$ in $W_\mathsf{Hyb}$ while causing the parties in $S_\H$ to output $V'$, with this breaking the agreement property. Note that the adversary can perform this sort of attack with at least $p \cdot q > 0$ probability because $i$ has a randomness compatible with $\W_V$ (a randomness that $i$ has in some fault-free world with the input assignment $V$ where $T_{S_\H \leftrightarrow i} = \tau_{S_\H \leftrightarrow i}$) with at least $p$ probability, and independently of this the parties in $S_\H$ have a joint randomness compatible with $\W_{V'}$ with at least $q$ probability. 

    In $W_{\mathsf{Hyb}}$, the byzantine parties simulate $W_V$ towards $i$ and simulate $W_{V'}$ towards $S_\H$. That is, each byzantine party $j \not \in S_\H \cup \{i\}$ acts towards $i$ as if $j$ were a correct party in the world $W_V$ (by sending messages it would have sent in $W_V$ on the same rounds), and likewise acts towards the parties in $S_\H$ as if $j$ were a correct party in the world $W_{V'}$.

    Let $W_j = W_{V'}$ for all $j \in S_\H$ and let $W_i = W_V$. We claim that a correct party $j \in S_\H \cup \{i\}$ can never distinguish $W_{\mathsf{Hyb}}$ from $W_j$. Since every correct party $j$ has the same input in $W_{\mathsf{Hyb}}$ and $W_j$, our claim is false only if there exists some first round $r \in [R]$ such that some correct party $j$ can distinguish $W_{\mathsf{Hyb}}$ from $W_j$ at the end of the round, even though no correct party $j'$ could distinguish $W_{\mathsf{Hyb}}$ from $W_{j'}$ at the beginning of the round. However, such a party $j$ cannot exist: If it did, it would have to receive different messages from another party $k$ in $W_{\mathsf{Hyb}}$ and $W_j$ on round $r$, even though this other party $k$ cannot exist. Finally, the reason why $k$ cannot exist is it can be neither correct nor byzantine. We show this below. \begin{itemize}
        \item The party $k$ cannot be byzantine because the byzantine parties behave identically towards $j$ in $W_{\mathsf{Hyb}}$ and $W_j$ on purpose.
        \item The party $k$ cannot be correct either. If it were so, then it would send the same round $r$ messages to $j$ in the worlds $W_{\mathsf{Hyb}}$ and $W_k$ since it would not have been able to distinguish the two worlds at the start of the round. If $W_k = W_j$, then $k$ is eliminated. If $W_k \neq W_j$, then either $k = i$ and $j \in S_\H$, or $k \in S_\H$ and $j = i$. Either way, $k$ sends the same round $r$ messages to $j$ in the worlds $W_k$ and $W_j$ (since the transcript between $k$ and $j$ is described by $\tau_{S_\H \leftrightarrow i}$ in both worlds), which again eliminates $k$.
    \end{itemize}

    The indistinguishability above means that in the world $\W_{\mathsf{Hyb}}$ the party $i$ outputs $V$ as it would in the world $W_V$, while the parties in $S_\H$ output $V'$ as they would in $W_{V'}$.
\end{proof}

\babbcorr

\begin{proof}
    Suppose $L = \Omega(Rn)$, and suppose there is an $R$-round byzantine broadcast protocol $\mathcal{BB}$ for $L$-bit inputs. Using $\mathcal{BB}$, the parties can achieve interactive consistency on $L$-bit inputs in $Rn$ rounds as follows: On the rounds $R(i-1)+1,\dots,Ri$, the party $i \in [n]$ broadcasts its input with $\mathcal{BB}$.\footnote{The reason why the parties do not broadcast their inputs in parallel is that if they were to do so, then it would no longer be clear to which byzantine broadcast instance each message belongs.} Observe that if $\mathcal{BB}$ violated the lower bound stated in Corollary \ref*{babbcorr}, then this interactive consistency protocol would violate the lower bound of Theorem \ref{interthm}. As Theorem \ref{interthm} is true, we conclude that the lower bound for byzantine broadcast holds.

    Next, let us consider byzantine agreement. Whenever byzantine agreement can be solved in $R$ rounds, byzantine broadcast can be solved in $R + 1$ rounds as follows: The sender sends its input to every other party in the first round, and in the rounds $2,\dots,R+1$ the parties reach byzantine agreement on the sender's input. Since the bit complexity of the first round is $L(n-1) = o(Ln^{1+1/d})$ and the one additional round is asymptotically irrelevant, the lower bound for byzantine broadcast implies the lower bound for byzantine agreement.

    Note that the reduction for byzantine agreement relies requires $n^{1/d} = \omega(1)$. For when $n^{1/d} = O(1)$, we recall that $\Z_\pj^{n,d}$ partitions the parties into two sets $U$ and $U'$ of size $\frac{n}{2}$ each, and assign the parties in $U$ a common uniformly random $L$-bit input $v$ while assigning the parties in $U'$ the fixed input $\mathbf{0}$. By the validity of byzantine agreement, the parties in $U$ must output $v$ (since as far as they know, the parties in $U'$ with inputs other than $v$ could all be byzantine), and this ensures agreement on $v$. So, the parties in $U'$ must learn $v$. In fault-free executions, for each party $i \in U'$ the bidirectional transcript $T_{([n] \setminus \{i\}) \leftrightarrow i}$ between $i$ and every other party must have at least $L$ bits of entropy;\footnote{The transcript entropy $H(T_{([n] \setminus \{i\}) \leftrightarrow i}) \geq L$ is required for $i$ to learn $v$. With less transcript entropy, it would not always be possible to deduce $v$ from $T_{([n] \setminus \{i\}) \leftrightarrow i}$; however, such a deduction must always be possible since $i$ must compute $v$ based on $T_{([n] \setminus \{i\}) \leftrightarrow i}$ (and on its initial state, which is unrelated to $v$).} hence, by Lemma \ref{sync-cov-lemma}, the communication between $i$ and every other party must be at least $\Omega(L)$ bits in expectation. This implies an \linebreak $\Omega(\frac{L|U'|}{2}) = \Omega(Ln) = \Omega(Ln^{1+1/d})$ total communication complexity bound if $n^{1/d} = O(1)$. 
\end{proof}

\section{Asynchronous Lower Bounds} \label{asyncsection}

Now, we switch over to asynchronous networks. This time, we have three similar but separate lower bound theorems to prove, as the asynchronous setting is less amenable to reductions. The theorems will make use of Lemma \ref{termlemma} (stated below), which we prove in the \hyperref[termlemma*]{appendix}. Intuitively speaking, the lemma says that if the adversary uses a message scheduling strategy which forces a party $i$ to learn the correct random output $X$ of a protocol solely based on the messages it receives from a fixed set of parties $S$, then the transcript of the communication directed from $S$ to $i$ must have $\Omega(H(X))$ bits of entropy for $i$ to learn $X$.

\apxthm{lemma}{termlemma}{
    Let $X$ be some finite uniformly random variable. Suppose that in a randomized asynchronous $n$-party protocol, some party $i$ must output a prediction $Y_i$ of $X$ solely based on the messages it receives from some fixed set of parties $S \subseteq [n] \setminus \{i\}$, and this prediction must satisfy $\Pr[Y_i = X] \geq 0.5 + \varepsilon$ for some absolute constant $\varepsilon > 0$. If $i$'s initial knowledge prior to running the protocol (its randomness and its input if it has any) is independent of $X$, then the one-way transcript $T_{S \rightarrow i}$ describing the communication directed from $S$ to $i$ must have at least $\varepsilon'H(X)$ bits of Shannon entropy for some absolute constant $\varepsilon' > 0$.
}

Our simplest lower bound in the asynchronous setting is the one for reliable broadcast.

\rbthm

\begin{proof}
    Recall that $\Z_\pj^{n,d}$ partitions the parties into two sets $U, U'$ of size $\frac{n}{2}$, where a quorum is a hyperplane of $U$. For the reliable broadcast lower bound, we assume that the sender $s$ is in $U'$. The input distribution $\V$ assigns a uniformly random $L$-bit input to $s$, which means that \linebreak the correct parties must output $\V$ if $s$ is correct. This distribution has $L$ bits of entropy.

    Consider a fault-free protocol execution. We suppose that the send-omission adversary employs a message scheduling strategy where the traffic between the parties in $U \cup \{s\}$ always flows freely while all other communication is blocked until when the parties in $U \cup \{s\}$ all terminate. This is a permissible way for the adversary to schedule messages because the parties in $U \cup \{s\}$ will indeed all terminate solely based on the messages they send each other. This is because they will not be able to distinguish the fault-free setting from the setting where the parties in $U' \setminus \{s\}$ are all silent faulty parties, while the remaining parties are correct. In this setting, the parties in $U \cup \{s\}$ would have to terminate due to the protocol's liveness (since the sender $s$ is correct); so, they must terminate in the fault-free setting too.

    Fix a quorum $S_\H \in \S_\pj^{n,d}$ and a party $i \in U' \setminus \{s\}$. Consider the random one-way transcript variable $T_{S_\H \rightarrow i}$ which describes the communication that the party $i$ receives from the parties in $S_\H$ in a random fault-free execution (where $s$'s input and the parties' randomnesses are all random) where the adversary begins the message scheduling as described above. We claim that $H(T_{S_\H \rightarrow i}) \geq \varepsilon'H(\V) = \varepsilon'L$ for some absolute constant $\varepsilon' > 0$. If this claim is true, then the parties in $S_\H$ must in expectation send $\Omega(L)$ bits in total to $i$ (by Lemma \ref{async-cov-lemma}), and since this is true for all $S_\H \in \S_\pj^{n,d}$ and $i \in U' \setminus \{s\}$ (where $U' \setminus \{s\}$ is a set of $\frac{n}{2} - 1 = \Theta(n)$ parties), Lemma \ref{comb-lemma} implies that the protocol involves $\Omega(Ln^{1+1/d})$ bits of expected communication.

    It remains to show that the claim $H(T_{S_\H \rightarrow i}) \geq \varepsilon'H(\V)$ is true. Suppose otherwise, for the sake of contradiction. Then, the adversary (who has made the parties in $U$ terminate before $i$ receives any message) can continue its message scheduling strategy by first delivering the messages from $S_\H$ to $i$, while delaying the other pending messages sent to $i$ for later. If the adversary does this, then the party $i$ will terminate solely based on the messages it receives from $S_\H$ without waiting for any other parties' messages, and its random output $Y$ will have to equal the sender's uniformly random $L$-bit input with at least $0.5 + \varepsilon$ probability by the validity of the reliable broadcast protocol (which must hold with at least the same probability); hence, Lemma \ref{termlemma} requires $H(T_{S_\H \rightarrow i}) \geq \varepsilon'H(\V)$ to hold. The reason $i$ will not wait for any other parties' inputs is that if it had a positive probability $p$ of doing so, then the adversary could force it to wait forever with the probability $p$ by causing the parties outside $S_\H \cup \{i\}$ to suffer from send-omission faults which prevent their messages from reaching the parties in $U'$, while keeping everything else the same. Then, $i$ waiting forever as we have described would be a \linebreak liveness violation because the parties in $S_\H$ would be correct parties who have terminated.
\end{proof}

Theorem \ref{rbthm} implies an $\Omega(Ln^2)$ bit expected communication complexity lower bound for any terminating reliable broadcast protocol for send-omission adversaries that can corrupt up to $n - 1$ parties, because such an adversary is stronger than a send-omission $\Z_\pj^{n,d}$-adversary. Note that omission-tolerant reliable broadcast is possible against up to $t < n$ send-omission faults with a simple echo-based protocol \cite{ht94}, as long as uniformity (the additional guarantee that if a faulty party outputs, then every correct party outputs) is not required. Requiring uniformity against send-omission faults decreases the optimal fault tolerance to $t < \frac{n}{2}$ \cite{cachin2011introduction}.

Now, we are ready to prove the lower bound for asynchronous byzantine agreement. We would like to get this lower bound using the same broadcast-to-agreement reduction we used to prove Corollary \ref{babbcorr}. Sadly, we cannot do this, because if a faulty sender sends its input to only some of the correct parties, then only some of the correct parties will have asynchronous \linebreak byzantine agreement inputs (the input they got from the sender), and this might lead to only some of the correct parties terminating the asynchronous byzantine agreement protocol. One way to solve this issue is to use an asynchronous byzantine agreement protocol that has the ``strong termination'' property defined in \cite{mw24}, which is comparable to the totality guarantee of reliable broadcast. However, we prefer to use standard termination guarantees and give a separate proof for the asynchronous byzantine agreement lower bound.

Note that some parts of our proof of Theorem \ref*{abathm} below mirror our proof of Theorem \ref{rbthm}. Hence, to avoid redundancy, we will treat these parts more concisely.

\abathm

\begin{proof}
    Again, the adversary structure $\Z_\pj^{n,d}$ partitions the parties into two sets $U, U'$ of size $\frac{n}{2}$. We choose an input distribution $\V$ that assigns a common uniformly random $L$-bit input $v$ to the parties in $U$ and the fixed input $\mathbf{0}$ to the parties in $U'$. As we did to prove Theorem \ref{rbthm}, we will lower bound the communication complexity of fault-free executions. The adversary employs a scheduling strategy that causes fault-free executions to work as follows: \begin{enumerate}
        \item The adversary allows the traffic between the parties in $U$ to flow freely, while blocking all other communication until the parties in $U$ all terminate. The parties in $U$ will terminate the protocol with the output $v$ because as far as they are aware, the parties in $U'$ could be faulty, and if this were the case, the validity and termination guarantees of the protocol would require the parties in $U$ to terminate with the output $v$.
        \item Now that the parties in $U$ have terminated with the output $v$, the agreement property of the protocol requires the parties in $U'$ to terminate with this output as well. Therefore, the transcript $T_{S_\H \rightarrow i}$ for any $S_\H$ and $i \in U'$ must have at least $\varepsilon'L$ bits of entropy for some absolute constant $\varepsilon' > 0$. This implies (by Lemmas \ref{termlemma} and \ref{async-cov-lemma}) that $S_\H$ has to send $\Omega(L)$ bits in expectation to $i$, and hence by Lemma \ref{comb-lemma} we conclude that the protocol must involve $\Omega(Ln^{1+1/d})$ bits of expected communication. The reason why $H(T_{S_\H \rightarrow i}) \geq \varepsilon'L$ is that $i$ must terminate the protocol with a prediction $Y_i$ of the uniformly random $L$-bit value $v$ solely based on the messages it receives from $S_\H$, in case everyone outside $S_\H \cup \{i\}$ is faulty, and this prediction must be correct with at least $0.5 + 2\varepsilon$ probability. We explain below how we obtain the probability $0.5 + 2\varepsilon$.
    \end{enumerate}
    
    The new subtlety with asynchronous byzantine agreement that was not present for reliable broadcast is that on the last step, the protocol achieving output safety with at least $0.75 + \varepsilon$ probability only implies that $i$'s prediction $Y_i$ of $v$ must equal $v$ with $0.5 + 2\varepsilon$ probability. The reason is that there are two error events $E_\mathsf{ag}$ and $E_\val$ this time, both of which can occur with up to $1 - (0.75 + \varepsilon) = 0.25 - \varepsilon$ probability, and we have $\Pr[Y_i = v] \geq 1 - \Pr[E_\mathsf{ag} \cup E_\val] \geq 0.5 + 2\varepsilon$ by the union bound. The first event $E_\mathsf{ag}$ represents an agreement violation. This is the bad event that in the fault-free execution, $i$'s output $Y_i$ differs from the output of some party in $U$, and this error event must happen with at most $0.25 - \varepsilon$ probability. Meanwhile, the event $E_\val$ represents a validity violation: This is the bad event that there exists some party in $U$ who does not output $v$ in the fault-free execution. We have $\Pr[E_\val] \leq 0.25 - \varepsilon$ because $E_\val$ corresponds to a validity violation in the alternate setting where only the parties in $U$ are correct (with this making $v$ the valid output they have to obtain) while the remaining parties are all silent faulty parties. The parties in $U$ cannot distinguish this alternate setting from the fault-free setting until after they all output in the fault-free setting, which means that in both settings they must all output $v$ except with probability at most $0.25 - \varepsilon$.
\end{proof}

Finally, the proof of the core set agreement lower bound follows the template of the proof above, but we use the adversary structure $\Z_\tpj^{n,d}$ and add an extra scheduling step to ensure that the parties agree on a core set that contains $\Theta(n)$ parties (and thus learn $\Theta(Ln)$ bits).

\acsthm

\begin{proof}
    Recall that $\Z_\tpj^{n,d}$ partitions the parties into three sets $U_1, U_2, U'$ of size $\frac{n}{3}$, and admits \linebreak a minimal quorum $C \supseteq U_1$ such that $[n] \setminus C$ is a maximal set in $\Z_\tpj^{n,d}$. We consider the input distribution $\V$ which assigns every party an i.i.d.\ uniformly random $L$-bit input. Once again, we want to lower bound the expected communication complexity of fault-free executions. In these executions, the adversary employs the scheduling strategy described below.  
    
    Initially, the adversary lets the traffic between the parties in $C$ flow freely, while blocking all other communication. The parties in $C$ will all terminate since they cannot afford to wait to hear from the parties outside, who could all be faulty silent parties as far as the parties in $C$ can tell. Moreover, each party $i \in C$ will output $Y_i = \{(c_1, y_1), \dots, (c_{|C|}, y_{|C|})\}$, where $C = \{c_1, \dots, c_{|C|}\}$ and $y_i$ is the input of the party $c_i$ for each $c_i \in C$. We show below that this claim concerning $Y_i$ is true with at least $0.625 + 3\varepsilon$ probability, based on the fact that the probability of a validity violation is at most $0.125 - \varepsilon$. To prove this claim, we define a validity violation event $E_\val$, define the three error events $E_a,E_b,E_c$ that could cause $Y_i$ to differ from $\{(c_1, y_1), \dots, (c_{|C|}, y_{|C|})\}$, and show that $\Pr[E_a \cup E_b \cup E_c] \leq 3\Pr[E_\val] \leq 0.375 - 3\varepsilon$. \begin{enumerate}
            \item Let $E_a$ be the bad event that for some $j \not \in C$ and any $y$ the set $Y_i$ contains the tuple $(j, y)$. \linebreak The party $i \in C$ decides on $y$ before it learns anything about $j$'s uniformly random input $v$, which means that $\Pr[y \neq v] = 1 - 2^{-L} \geq 0.5$. Hence, we have $0.125 - \varepsilon \geq \Pr[E_\val] \geq \Pr[E_a \cap E_\val] \geq 0.5\Pr[E_a]$.
            \item Let $E_b$ be the bad event that there is a party $c \in C$ such that $Y_i$ does not contain a tuple of the form $(c, y)$. Observe that $E_b \setminus E_a$ is impossible. The reason is that if $E_b$ happens but $E_a$ does not, then for some proper subset $C'$ of $C$, the set $Y_i$ only contains tuples of the form $(c', y')$ when $c' \in C'$. This is a contradiction because core set agreement requires $[n] \setminus C'$ to be in $\Z_\tpj^{n,d}$, which contradicts the maximality of $[n] \setminus C$ in  $\Z_\tpj^{n,d}$.
            \item Let $E_c$ be the bad event that there exists any tuple $(j, y) \in Y_i$ such that $j$'s input is not $y$. This is a validity violation; hence, $\Pr[E_c] \leq \Pr[E_\val] \leq 0.125 - \varepsilon$.
        \end{enumerate}
        Combining all of the inequalities above, we get $\Pr[E_a \cup E_b \cup E_c] = \Pr[E_a] + \Pr[E_b \setminus E_a] + \Pr[E_c \setminus (E_a \cup E_b)] \leq \Pr[E_a] + 0 + \Pr[E_c] \leq 0.375 - 3\varepsilon$.
    
    In the next message scheduling stage, the adversary allows the traffic between the parties in $U = U_1 \cup U_2$ to flow freely, but it keeps delaying any messages whose senders or intended recipients are in $U'$. This will cause the parties in $U$ to terminate. As usual, the quorum $U$ cannot wait to hear from the parties in $U'$ because the parties in $U'$ could all be faulty.

    Now that the parties in $U$ have terminated with the output, we are in a situation analogous to one we faced in our asynchronous byzantine agreement lower bound proof. That is, for each quorum $S_\H \subseteq U$ (which consists of the parties in $U_1$ and $U_2$ whose labels in $PG(d,q)$ are incident to the hyperplane $\H$) and each party $i \in U'$, the one-way transcript $T_{S_\H \rightarrow i}$ must have $\Omega(Ln)$ bits of entropy for $i$ to terminate with the correct output $Y = \{(c_1, y_1), \dots, (c_{|C|}, y_{|C|})\}$ (which is a uniformly random variable that has $L|C| = \Omega(Ln)$ bits of entropy) with at least $0.5 + 4\varepsilon$ probability, which implies (by Lemmas \ref{termlemma} and \ref{async-cov-lemma}) that the parties in $S_\H$ are required to send $\Omega(Ln)$ bits in expectation to $i$, and hence by Lemma \ref{comb-lemma} that the protocol must involve $\Omega(Ln^{2+1/d})$ bits of expected communication. The reason for the $0.5 + 4\varepsilon$ probability here is that if $i$ outputs anything but $Y$, then one (or both) of two error events must have happened: Either some party $c \in C$ obtained some output other than $Y$ from the protocol (which we argued must happen with at most $0.375 - 3\varepsilon$ probability), or $i$'s output disagrees with the output of some party $c \in C$ (which can happen with at most $0.125 - \varepsilon$ probability). 
\end{proof}

\subsection{The Importance of Termination} \label{pullcastsection}

In this section, we prove that our lower bound for terminating reliable broadcast no longer holds if we drop the termination requirement and allow the parties to run forever after they output. We will then discuss why this is also the case for the core set agreement and the asynchronous byzantine agreement lower bounds, and even if the adversary is byzantine.

For our proof, we design a deterministic and error-free asynchronous protocol $\PullCast$ for reliable broadcast against any general-omission adversary. For any choice of a parameter $\delta > 1$ (which is a multiple of $\frac{1}{n}$), $\PullCast$ can support $L$-bit inputs with $O(\delta n^2)$ messages and $(1 + \frac{1 - 1/n}{\delta - 1})Ln + O(\delta n^2\log(n\delta))$ bits of communication, albeit with the extremely high latency $2(\delta - 1)n - 1$.\footnote{The latency $2(\delta - 1)n - 1$ means that in a $\PullCast$ execution where the adversary delays the correct parties' messages by at most $1$ time unit, if a correct party outputs from the protocol (which is what the sender immediately does if it is correct) at some point in time $T$, then every correct party outputs from the protocol by the time $T + 2(\delta - 1)n - 1$. So, if $\delta - 1 = \Theta(1)$, which is required for $\PullCast$ to achieve the communication complexity $O(Ln + n^2\log n)$, then the latency is $\Theta(n)$.} This complexity does not violate our lower bounds as $\PullCast$ does not terminate. That is, the parties keep running even after they learn the sender's input.

To obtain $\PullCast$, we need to use an $(N, k)$-erasure correcting code. Such a code encodes a string $m$ into a vector of $N$ symbols $\mathbf{s} = (s_1,\dots,s_N)$ such that any $k$ of these symbols are enough to reconstruct $m$. Formally, we have two functions $\mathsf{Enc}_{N, k}(m)$ and $\mathsf{Dec}_{N, k}(\tilde{\mathbf{s}})$. The encoding function $\mathsf{Enc}_{N, k}(m)$ maps a string $m$ into a vector of $N$ symbols $\mathbf{s} = (s_1, \dots, s_N)$. The decoding function $\mathsf{Dec}_{N, k}(\tilde{\mathbf{s}})$ takes in an incomplete symbol vector $\tilde{\mathbf{s}} = (\tilde{s}_1, \dots, \tilde{s}_N)$ with some missing symbols (represented by $\bot$) and outputs the decoded string $m$, provided that $\tilde{\mathbf{s}}$ contains at least $k$ non-$\bot$ symbols and each non-$\bot$ symbol $\tilde{s}_j$ is the $j^{\text{th}}$ symbol of $\mathsf{Enc}_{N, k}(m)$.

We use Reed-Solomon $(N,k)$-erasure codes \cite{reed-solomon}, which encode $L$-bit strings into $N$ symbols of size $\Theta(\frac{L}{k} + \log N)$ each. Given the parameter $\delta > 1$ such that $\delta n$ is an integer, $\PullCast$ uses the erasure code parameters $(N, k) = (\delta n - 1, (\delta - 1)n)$.

In $\PullCast$, a party $i$ who knows the sender's input $v$ (initially only the sender $s$) computes $(s_1, \dots, s_{\delta n - 1}) = \mathsf{Enc}_{\delta n - 1, (\delta - 1)n}(v)$, sends the message $\msg{\SYM,i,s_i}$ containing the $i^{\text{th}}$ symbol to every other party, and starts answering symbol requests. That is, if $i$ receives (or has received before learning $v$) a symbol request message $\msg{\REQ,k}$ from a party $j$ for the $k^{\text{th}}$ symbol, then $i$ replies to $j$ with the answer $\msg{\SYM,k,s_k}$. A party $i$ outputs $v$ when it knows $v$, which is when \linebreak $i$ begins the protocol if $i$ is the sender and when $i$ receives $(\delta - 1)n$ symbols otherwise.

To support long inputs efficiently, $\PullCast$ uses a dynamic symbol request system. Unlike most other protocols based on Cachin and Tessaro's seminal work on fault-tolerant information dispersal \cite{ct05}, $\PullCast$ does not fix in advance which party should send a symbol to whom. Instead, each party $i$ maintains a counter $c_i$ of how many symbols it has received so far, so that whenever $i$ receives any new symbol from any party $j$, it can increment $c_i$ and request the next symbol $s_{n + c_i}$ from $j$. This way, $i$ requests more symbols from the parties that reply quickly to $i$. We note that this kind of dynamic request pattern is well known and used in practice by peer-to-peer file sharing protocols like BitTorrent, even though it is uncommon in the fault-tolerant consensus literature.

\begin{theorem}
    $\PullCast$ is a non-terminating reliable broadcast protocol which can tolerate any general-omission adversary with perfect security.
\end{theorem}

\begin{algobox}[\textbf{ with the parameter }$\delta > 1$]{$\PullCast$}
    \State $c_i \gets 0$ and $(\tilde{s}_1, \dots, \tilde{s}_{\delta n - 1}) \gets (\bot, \dots, \bot)$

    \Upon{learning the sender's input $v$ (initially as the sender or via decoding)}
        \State output $v$ but do not terminate
        \State $(s_1, \dots, s_{\delta n - 1}) \gets \mathsf{Enc}_{\delta n - 1, (\delta - 1)n}(v)$
        \State send $\msg{\SYM, i, s_i}$ to all parties $i \in [n]$
    \EndUpon

    \Upon{receiving a symbol request $\msg{\REQ, k}$ from a party $j$}
        \State send $\msg{\SYM, k, s_k}$ to $j$
    \EndUpon

    \Upon{receiving a symbol $\msg{\SYM, k, s}$ from a party $j$}
        \If{you do not know $v$}
            \State $c_i \gets c_i + 1$
            \State $\tilde{s}_k \gets s$
            \If{$c_i = (\delta - 1)n$}
                \State $v \gets \mathsf{Dec}_{\delta n - 1, (\delta - 1)n}(\tilde{s}_1, \dots, \tilde{s}_{\delta n - 1})$ \label{line:decode}
                \State trigger the ``learning the sender's input $v$'' block
            \Else
                \State send $\msg{\REQ, n + c_i}$ to $j$ \label{line:request}
            \EndIf
        \EndIf
    \EndUpon
\end{algobox}

\begin{proof}
    Suppose the sender's input is $v$, with $(s_1, \dots, s_{\delta n - 1}) = \mathsf{Enc}_{\delta n - 1, (\delta - 1)n}(v)$. The sender outputs $v$ immediately when it acquires its input, and we claim that the other parties can also only output $v$. Suppose not, for the sake of contradiction. Then, there exists a first non-sender party $i$ who obtains a different output $v'$ by computing $\mathsf{Dec}_{\delta n - 1, (\delta - 1)n}(\tilde{s}_1, \dots, \tilde{s}_{\delta n - 1})$  on Line \ref*{line:decode}. At the moment when this happens, it cannot be true that $\tilde{s}_j \not \in \{s_j, \bot\}$ for any $j$, because if $\tilde{s}_j \neq \bot$, then $i$ has set $\tilde{s}_j \gets s_j'$ after receiving the message $\msg{\SYM, j, s_j'}$ from some party $k$, and $s_j' = s_j$ because $k$ is a party who knows the correct output $v$. So, it must be the case that the vector $(\tilde{s}_1, \dots, \tilde{s}_{\delta n - 1})$ contains less than $(\delta - 1)n$ non-$\bot$ symbols. However, because $i$ only reaches  Line \ref*{line:decode} after performing $(\delta - 1)n$ symbol updates and $i$ never updates a symbol twice, this is also impossible; so, we reach the desired contradiction. The reason why $i$ never updates the same symbol $\tilde{s}_k$ twice is that for each $k \in [\delta n - 1]$, there can only be one party $j$ who sends $i$ the message $\msg{\SYM, k, s_k}$. If $k \leq n$, then this party $j$ is $k$, and if $k > n$, then $j$ is the party to whom $i$ sent the message $\msg{\REQ, n + c_i}$ on Line \ref*{line:request} when $i$ had \nolinebreak $c_i = k - n$.

    It remains to prove liveness. We do so by showing that if there is any correct party $i$ (the sender or not) who knows the sender's input $v$, then every correct party learns $v$. Suppose such a party $i$ exists, and consider any correct party $j$ that does not know $v$. The party $i$ sends $j$ the symbol $s_i$, and this triggers a feedback loop where $j$ can continually request a new symbol from $i$ and receive it from $i$ until when $j$ acquires $(\delta - 1)n$ symbols (not necessarily all from $i$) and outputs $v$. Note that this argument also shows that the latency is $2(\delta - 1)n - 1$, as it takes one hop for $i$ to send the symbol $s_i$ to $j$, and then it takes $(\delta - 1)n - 1$ round trips between $j$ and $i$ (each taking $2$ hops) for $j$ to receive $(\delta - 1)n - 1$ more symbols. If $j$ receives any symbols from parties other than $i$, then this will only reduce the number of round trips that must take place between $j$ and $i$ for $v$ to output; therefore, the latency will not go up.

    Lastly, note that we do not have to worry about any correct party $i$ failing to respond to a symbol request due to it not knowing the symbol. If any party $j$ requests a symbol from $i$, then $j$ must have received some symbol from $i$ before doing so, and this means that $i$ already knows $v$ because the first time when $i$ sends a symbol to $j$ is when $i$ learns $v$.
\end{proof}

\subparagraph{The Complexity of $\PullCast$.} We have already looked at $\PullCast$'s latency and shown that it is $2(\delta - 1)n - 1$. The communication in $\PullCast$ amounts to at most $\delta n - 1$ symbol ($\SYM$) messages of size $\frac{L}{(\delta - 1)n} + O(\log(\delta n))$ in addition to at most $(\delta - 1)n$ request ($\REQ$) messages of size $O(\log(\delta n))$ per party; so, $\PullCast$'s total communication complexity is $O(\delta n^2)$ messages and $(\delta n - 1)(\frac{L}{(\delta - 1)n})n + O(\delta n^2\log(\delta n)) = (1 + \frac{1 - 1/n}{\delta - 1})Ln + O(\delta n^2\log(\delta n))$ bits.

\subparagraph{Byzantine Fault Tolerance.} Since $\PullCast$ can only tolerate send-omission adversaries, a natural question is: Can it be made byzantine fault tolerant, assuming the adversary structure satisfies the necessary $\Q^3$ condition? While a formal treatment of byzantine fault tolerant reliable broadcast protocol is outside the scope of this paper (as we are more interested in the lower bounds), below we sketch how we believe one could harden $\PullCast$ against byzantine faults. Note that the changes which we propose increase the communication complexity to $(1 + \frac{1 - 1/n}{\delta - 1})Ln + O(\delta n^2\kappa\log(\delta n) + \delta n^3\log(\delta n))$ bits, where $\kappa$ is size of a cryptographic hash.
\begin{itemize}
    \item \textbf{Preventing Duplicated Requests.} The first new challenge that emerges with byzantine faults is that a byzantine party can request a symbol from multiple correct parties, to waste communication by making multiple correct parties send it the symbol. To prevent this, we could replace each symbol request where a party $i$ sends the message $\msg{\REQ, k}$ to a party $j$ with an $O(n)$-times costlier fault-tolerant request mechanism that prevents $i$ (if it is byzantine) from requesting the $k^\text{th}$ symbol from two distinct parties.\footnote{In this mechanism, $i$ would send everyone the message $\msg{\REQ, i, k, j}$ to indicate that it requests the $k^\text{th}$ symbol from the party $j$, and everyone would relay this request to $j$. The party $j$ would send the symbol upon hearing the request relayed from a quorum $S_{i,k,j}$, where $[n] \setminus S_{i,k,j} \in \Z$. To prevent $i$ from causing two distinct parties $j$ and $j'$ to send it the same symbol, we would make each correct party only relay the first request $i$ makes for each symbol index $k$. This way, the quorums $S_{i,k,j}$ and $S_{i,k,j'}$ required for $j$ and $j'$ to both send $i$ the same symbol cannot both form: Such two quorums would have to contain a common correct party by the $\Q^3$ condition, but a correct party does not relay duplicated requests.} Unfortunately, this would increase the $O(\delta n^2\log(\delta n))$ bit complexity of the requests to $O(\delta n^3\log(\delta n))$.
    \item \textbf{Handling Incorrect Symbols.} The second challenge is that the byzantine parties can send the other parties incorrect or inconsistent symbols. To prevent this, we would require the parties who know the sender's input $v$ to send symbols together with cryptographic proofs which witness the correctness of the symbol. This can be implemented with Merkle tree proofs \cite{merkle} of length $O(\kappa \log(\delta n))$, where the Merkle tree is computed over the correct symbol list $(s_1, \dots, s_{\delta n - 1})$. To establish the basis for these proofs, the sender would also \linebreak have to reliably broadcast a Merkle root hash for the proofs to be checked against. Finally, a non-sender party would obtain an output $y$ from the protocol after it outputs some root hash $h$ from the sender's root hash reliable broadcast, it acquires $(\delta - 1)n$ symbols which are correct with respect to $h$, it decodes these symbols into the tentative output $y$, and it verifies that the sender computed $h$ correctly by recomputing it based on $y$.
\end{itemize} We note that our suggestion of using Merkle tree proofs to handle incorrect or inconsistent erasure code symbols is very standard in the byzantine fault tolerance literature \cite{ct05,nayak20}.

\subparagraph{Byzantine Agreement and Core Set Agreement.} While $\PullCast$ only proves that our reliable broadcast lower bound requires termination, our bounds for asynchronous byzantine agreement and core set agreement also need termination. Below, we sketch the reasons why. \begin{itemize}
    \item \textbf{Core Set Agreement.} By \cite{bcg93}, one can reduce core set agreement to $n$ instances of reliable broadcast and $n$ instances of binary byzantine agreement, and this reduction also works against general adversaries \cite{acc25}. Therefore, one can compose $n$ instances of $\PullCast$ with $n$ instances of an asynchronous binary byzantine agreement protocol which has a complexity that does not depend on $L$ (e.g.\ \cite{c23}) to get a core set agreement protocol that can support $L$-bit inputs with $O(Ln^2)$ bits of communication for all sufficiently large $L$.
    \item \textbf{Byzantine Agreement.} For byzantine agreement, we observe that $\PullCast$ can support multiple senders who have a common input $v$ and behave as parties who know the correct output $v$, with the guarantee that if any sender is correct or some correct party outputs, then every correct party outputs. Hence, if the parties do not have to terminate, they can reach byzantine agreement on their $L$-bit inputs with $O(Ln + \dots)$ bits of communication by 1) running a byzantine agreement on their input's cryptographic hashes and reaching intrusion-tolerant byzantine agreement on the hash $h$ of a correct party's input\footnote{Intrusion tolerance \cite{mr10} means that the parties agree on either a correct party's input or a safe default value $\bot$. The latter can only happen if the correct parties have non-matching inputs. If the parties agree on $\bot$ in our reduction, then they can output $\bot$ from the whole protocol instead of running $\PullCast$.}, and then 2) running a $\PullCast$ instance where the parties whose inputs have the hash $h$ behave \linebreak as the parties who know $v$ to make sure that everyone learns $v$. This strategy will succeed unless a hash collision occurs, which is unlikely by the hash function's collision resistance. We note that this sort of reduction (where $\PullCast$ serves as a reconstruction protocol) is standard in the literature \cite{nayak20}.
\end{itemize}

\subsection{Termination with \texorpdfstring{\mathversion{bold}$O(L_\mathsf{out} \cdot n^{1+1/d} + n^2\log n)$}{O(L\_out * n\textasciicircum(1+1/d)} Bits} \label{termsection}

In the previous subsection, we showed that send-omission tolerant non-terminating reliable broadcast admits a solution that costs $O(L_\mathsf{out} \cdot n)$ bits of communication whenever the output length $L_\mathsf{out}$ is sufficiently large, and sketched why this is also the case for core set agreement and asynchronous byzantine agreement. Now, we will show that using the $\Q^d$ assumption alone, we can achieve termination with $O(L_\mathsf{out} \cdot n^{1+1/d})$ bits of further communication. Thus, we will show that our asynchronous lower bounds are tight, not just for the specific adversary structure families $\Z_\pj^{n,d}$ and $\Z_\tpj^{n,d}$ but for send-omission adversaries that are characterized by $\Q^d$-satisfying adversary structures in general. To obtain this result, we design a termination protocol $\Term$ that can add termination to any general-omission tolerant agreement protocol with $O(L_\mathsf{out} \cdot n^{1+1/d})$ bits of further communication (where $L_\mathsf{out}$ is the length of the output the parties agreed on), assuming the adversary is characterized by a $\Q^d$-satisfying adversary structure $\Z$. In particular, one can compose $\PullCast$ with $\Term$ to get a terminating reliable broadcast protocol for $L$-bit inputs which can tolerate any $\Q^d$ general-omission adversary with $O(Ln^{1 + 1/d} + n^2\log n)$ bits of communication. At the end of the section, we also sketch how $\Term$ can be made byzantine fault tolerant when $d \geq 3$.

\begin{proposition} \label{denseprop}
    For every $d \geq 1$ and every $\Q^d$-satisfying $n$-party adversary structure $\Z$, there exists a non-empty party set $G \subseteq [n]$ such that $|G \setminus Z| \geq |G|n^{-1/d}$ for all $Z \in \Z$.
\end{proposition}

\begin{proof}
    We prove the proposition with a greedy algorithm that finds the set $G$. Let $S_0 = [n]$. For each $i \in [d]$, we inductively define the sets $Z_i$ and $S_i$ as follows: $Z_i$ is any arbitrary set in $\Z$ which minimizes $|S_{i-1} \setminus Z_i|$, and $S_i = S_{i-1} \setminus Z_i$. Observe that by the $\Q^d$ condition we have $S_d = [n] \setminus (\bigcup_{i = 1}^d Z_i) \neq \emptyset$. This means that there have to exist two successive terms $|S_i|,|S_{i+1}|$ in the sequence $|S_0|,|S_1|,\dots,|S_d|$ such that $|S_{i+1}| \geq |S_i|n^{-1/d}$, as otherwise we would have $|S_d| < |S_0|(n^{-1/d})^d = 1$. The set $G$ we seek is the set $S_i$ such that $|S_{i+1}| \geq |S_i|n^{-1/d}$, since we have $|S_{i+1}| = \min_{Z \in \Z}|S_i \setminus Z|$ by definition.
\end{proof}

In $\Term$, we use the set $G$ from Proposition \ref*{denseprop} by tasking the parties in $G = \{g_1, \dots, g_{|G|}\}$ with helping every party terminate. That is, we encode the agreed upon output $v$ into $|G|$ symbols $(s_1, \dots, s_{|G|})$ with a $(|G|, \ceil{|G|n^{-1/d}})$-erasure code, and require each party $g_i \in G$ to send everyone the symbol $s_{g_i}$. Because $G$ contains at least $\ceil{|G|n^{-1/d}}$ correct parties, this suffices for every party to receive at least $\ceil{|G|n^{-1/d}}$ symbols and thus learn $v$. Finally, to obtain termination, we add totality: We make it so that before a party $i$ terminates $\Term$ with the output $v$, $i$ sends each symbol $s_{g_j}$ to the respective party $g_j$, and thus ensures that every correct party in $G$ will be able to fulfill its duty.

Formally, $\Term$ involves some other asynchronous protocol $\Pi$ which the other parties run. $\Pi$ here can be any protocol which guarantees that no two parties (correct or faulty) obtain distinct outputs from it.\footnote{To keep our lower bounds for asynchronous protocols general, we proved them using only the agreement property that the correct parties' outputs match. However, for $\Term$ to work, we also need $\Pi$ to achieve the uniform agreement guarantee that no two parties (correct or faulty) obtain non-matching outputs. Note that $\PullCast$ achieves this natural guarantee, as explained in its security proof.} $\Term$ is sequentially composed with $\Pi$, and it ensures the following totality property: If any correct party outputs from $\Pi$, then every correct party terminates $\Term$. Furthermore, there can be no spurious $\Term$ outputs: If a party outputs $v$ from $\Term$, then some party must have output $v$ from $\Pi$ earlier. Note that if any party terminates $\Term$, then it also terminates $\Pi$ (stops running it) whether or not it has output from $\Pi$, as this is the whole point of $\Term$. Depending on $\Pi$ (for example if $\Pi$ is $\PullCast$), this quitting can cause $\Pi$ to lose its liveness, preventing the other correct parties from outputting from $\Pi$, but $\Term$ still ensures that every correct party terminates with the correct output.

\begin{algobox}{$\Term$}
    \State $(\tilde{s}_{1}, \dots, \tilde{s}_{|G|}) \gets (\bot, \dots, \bot)$ and $\mathit{relayed} \gets \mathsf{false}$
    \Upon{learning the output $v$ (from $\Pi$ or via decoding)}
        \State $(s_{1}, \dots, s_{|G|}) \gets \mathsf{Enc}_{|G|, \ceil{|G|n^{-1/d}}}(v)$
        \For{each $g_j \in G$}
            \State send $\msg{\SYM, g_j, s_j}$ to $g_j$
        \EndFor
        \If{$i \in G$ and $\mathit{relayed} = \mathsf{false}$}
            \State send $\msg{\SYM, i, s_i}$ to all parties
        \EndIf
        \State terminate $\Term$ with the output $v$ and quit $\Pi$
    \EndUpon

    \Upon{receiving a message $\msg{\SYM, j, s}$}
        \If{$i \in G$ and $j = i$ and $\mathit{relayed} = \mathsf{false}$}
            \State send $\msg{\SYM, i, s}$ to all parties \label{line:relay}
            \State $\mathit{relayed} \gets \mathsf{true}$
        \EndIf
        \State $\tilde{s}_{j} \gets s$
        \If{the vector $(\tilde{s}_1, \dots, \tilde{s}_{|G|})$ contains exactly $\ceil{|G|n^{-1/d}}$ non-$\bot$ symbols}
            \State $v \gets \mathsf{Dec}_{|G|,\ceil{|G|n^{-1/d}}}(\tilde{s}_{1}, \dots, \tilde{s}_{|G|})$
            \State trigger the ``learning the output $v$'' block
        \EndIf
    \EndUpon
\end{algobox}

\begin{theorem}
    Suppose $\Pi$ is a protocol from which the parties can only obtain equal outputs. If a correct party outputs from $\Pi$, then every correct party terminates $\Term$, and a party can only output any $v$ from $\Term$ if some party (possibly the same party) output $v$ from $\Pi$ earlier.
\end{theorem}

\begin{proof}
    Let $v$ be the uniquely determined value such that the parties can only output $v$ from $\Pi$. If $v$ does not exist at some point in time (i.e.\ if $\Pi$ has not reached univalency), then no one \linebreak has output from $\Pi$. This prevents the parties from terminating $\Term$. The reason is that the first $\SYM$ message in a $\Term$ execution must be sent by a party who already knows the output $v$ it will terminate $\Term$ with, which means that if there is a first party who outputs some $v$ from $\Term$, then it decides on $v$ not via decoding $\SYM$ messages but via outputting $v$ from $\Pi$.

    Now, suppose that some correct party $i$ outputs $v$ from $\Pi$. Then, this party $i$ will compute  $(s_{1}, \dots, s_{|G|}) \gets \mathsf{Enc}_{|G|, \ceil{|G|n^{-1/d}}}(v)$ and send each symbol $s_j$ to the party $g_j$ for all $j \in [|G|]$, thus enabling the party $g_j \in G$ to send everyone the symbol $s_j$. Each party $g_j$ sends everyone the symbol $s_j$ (either once it receives $s_j$ from another party or once it computes $s_j$ just before \linebreak it outputs) and $G$ contains at least $\ceil{|G|n^{-1/d}}$ correct parties; consequently, everyone receives at least $\ceil{|G|n^{-1/d}}$ symbols from the parties in $G$. This is enough for every party to receive $\ceil{|G|n^{-1/d}}$ symbols, decode them into $v$ and finally terminate $\Term$ with the output $v$, unless the party terminates $\Term$ with the output $v$ earlier by outputting $v$ from $\Pi$.
\end{proof}

\subparagraph{The Complexity of $\Term$.} Let $L_\mathsf{out}$ be the length of the output $v$ that the parties obtain from $\Pi$. The communication in $\Term$ consists of at most $2n|G|$ $\SYM$ messages: $n|G|$ for each party $i \in [n]$ to inform each party $g_j \in G$ of the symbol $s_j$, and $n|G|$ for each party $g_j \in G$ \linebreak to inform everyone else of $s_j$. With Reed-Solomon erasure coding, the symbols are of size $O(\frac{L_\mathsf{out}}{|G|n^{-1/d}} + \log n)$; so, the total communication complexity is $O(L_\mathsf{out} \cdot n^{1+1/d} + n|G|\log n) = O(L_\mathsf{out} \cdot n^{1+1/d} + n^2\log n)$ bits. Finally, the latency is $2$; that is, if a correct party $i$ outputs from $\Pi$ at some time $T$ in an execution where the adversary delays each message by at most $1$ time unit, then every correct party outputs by the time $T + 2$. This is because with one hop the party $i$ informs every party $g_j \in G$ of the symbol $s_j$, and with a second hop the correct parties in $G$ all send everyone the symbols they need to learn the output and terminate.

\subparagraph{Byzantine Fault Tolerance.} $\Term$ tolerates general-omission adversaries. While a formal treatment of byzantine fault tolerance is out of scope, we sketch how a byzantine fault tolerant version of $\Term$ could work, assuming the $\Q^d$ condition for any $d \geq 3$. For this, we only need to introduce some mechanism to protect the parties from any incorrect/inconsistent symbols that the byzantine parties might propose, and as we discussed in the previous subsection, this can be done by attaching a Merkle proof to each symbol to prove the symbol correct. These proofs would be relative to the Merkle tree computed on the vector $(s_{1}, \dots, s_{|G|})$. To unanimously establish the root hash, the parties would also run a simple quadratic-complexity reliable agreement protocol (with this requiring the $\Q^3$ condition) in parallel, and terminate it before they terminate $\Term$. A party's reliable agreement input would be the root hash $h$ which it can compute upon outputting $v$ from $\Pi$.\footnote{Reliable agreement \cite{ddlmrs24} is a termination protocol like $\Term$. We would use it here to ensure that if some correct party terminates it with the root hash, then all of them terminate it with the root hash. While the reliable agreement protocol in \cite{ddlmrs24} is for $t < n/3$ threshold byzantine adversaries, it is easy to adapt it to the $\Q^3$ setting with a standard transformation \cite{kf05}: Replace any line which says ``upon receiving a message $m$ from $k$ parties'' with ``upon receiving $m$ from a set of $S$ parties such that $P_k(S)$ holds,'' where $P_k(S)$ is the predicate for whether $S \not \in \Z$ if $k = t + 1$, and for whether $[n] \setminus S \in \Z$ if $k = n - t$.}

\section{Discussion \& Future Work} \label{discuss}

In this work, we proved a general $\Omega(L_{\mathsf{out}} \cdot n^{1+1/d})$ bit communication complexity lower bound \linebreak against $\Q^d$ adversaries for a variety of fundamental agreement tasks; assuming the requirement of error-free security against byzantine adversaries in the synchronous setting and assuming termination against send-omission adversaries in the asynchronous setting.

\subparagraph{Synchronous Tightness.} An open question we leave for future work is whether the lower bounds we have proven for error-free synchronous protocols are tight (not necessarily just for $\Z_\pj^{n,d}$ but for $\Q^d$ adversaries in general), and if so, to come up with protocols that match them. While the combination of byzantine faults and error-free security might make it quite difficult to design such protocols, we conjecture that our lower bounds for synchronous protocols are tight like their asynchronous counterparts.

\subparagraph{Las Vegas Protocols.} To keep our proofs clear in the synchronous setting, we assumed a fixed round complexity of $R$. However, our entropy-based arguments extend to all error-free synchronous protocols no matter the input length $L \geq 1$ or the round complexity, even for Las Vegas style  protocols with variable round complexities. That is, one could by adapting our lower bound proofs show that the expected sum $\E[\sum_{i = 1}^n\sum_{j = i + 1}^nH(T_{i \leftrightarrow j})]$ (where $T_{i \leftrightarrow j}$ is the transcript describing the communication between $i$ and $j$) is $\Omega(Ln^{2+1/d})$ for interactive consistency and $\Omega(Ln^{1+1/d})$ for byzantine agreement/broadcast, no matter $L$ or the round complexity.\footnote{Note that the reduction from byzantine broadcast to interactive consistency (Corollary \ref{babbcorr}) breaks if the byzantine broadcast protocol does not have a fixed round complexity $R$, as our strategy to cleanly separate the broadcast instances was to allot a block of $R$ rounds to each instance. To address this issue without needing protocol ID tags, we can instead employ a round-robin schedule, where for each $k \in [n]$ \linebreak the parties run the $k^{\text{th}}$ byzantine broadcast instance in a round $r$ iff $r \equiv k \pmod n$.} Less abstractly, these are the correct communication complexity lower bounds for any error-free synchronous protocol in a tagged communication model where each message is tagged with the pertinent metadata (sender/receiver IDs and a round number), no matter the round complexity. Intuitively speaking, the reason is that when these tags are included in the messages, a transcript $T_{S \leftrightarrow i}$ between a party set $S$ and a party $i$ becomes just a message sequence, and a lower bound of $H(T_{S \leftrightarrow i}) \geq h = \Omega(1)$ for the transcript entropy also implies a lower bound of $\Omega(h)$ for the expected total length of the messages.

\subparagraph{Termination with High Probability.} Another class of protocols which merit attention are asynchronous protocols that terminate with high probability, but not probability $1$. While we chose not to formalize these protocols to avoid greatly complicating the math, we conjecture that our lower bounds for asynchronous protocols generalize to them. A way to adapt our proofs to show this formally could be to use events which track if some parties or party sets terminate as they should, and only lower bound the communication complexity of the good executions where the events all occur. One would also have to condition the input distribution on the events occurring and account for the resulting input distribution entropy loss, as the posterior distribution of the inputs when the parties terminate may be non-uniform.

\subparagraph{Balanced Communication.} An exciting question which we did not consider in this work is how balanced the communication can be. In the threshold adversary setting, it is considered good for a protocol to be balanced. Such a protocol ensures that if there are $c$ correct parties in an execution, then each of them shoulders a $O(\frac{1}{c})$ proportion of the total communication complexity. In non-symmetric protocols (where the parties take on different roles), achieving balance can be quite challenging. For example, a balanced reliable broadcast protocol which costs $O(Ln + n^2\log n)$ bits \cite{long22} cannot allow the sender to just send its input to everyone else. We are curious about when general adversary tolerant protocols (and protocols like $\PullCast$ which can tolerate $t < n$ faults) can be balanced. As far as we are aware, this question has not been studied before. The answer could depend on whether the network is synchronous or not, on the adversary, and on whether termination is required if the network is asynchronous.

\subparagraph{Crash Faults.} A question one may ask is: Do our lower bounds for terminating asynchronous protocols also hold against crash faults? The answer is no, at least not for reliable broadcast, if the adversary must deliver the messages that a crashing party sent before crashing. To show this, we sketch a simple terminating reliable broadcast protocol that costs $L(n-1) + O(n^2)$ bits and can tolerate any number of crash faults. The sender sends everyone its $L$-bit input, and after this it reliably broadcasts the string $\texttt{OK}$ with a reliable broadcast instance $\mathcal{RB}_\mathsf{ok}$ (using the echo-based protocol in \cite{ht94}) to signal to the other parties that it did not crash before it sent everyone its input. Then, a party can terminate the full protocol after terminating $\mathcal{RB}_\mathsf{ok}$ and receiving the sender's message. If the sender is correct, then every correct party receives its input and terminates $\mathcal{RB}_\mathsf{ok}$. If the sender crashes before it can send everyone its input, then it never initiates the broadcast $\mathcal{RB}_\mathsf{ok}$, and therefore the parties never terminate $\mathcal{RB}_\mathsf{ok}$. And finally, if the sender crashes after it sends everyone its input, then the correct parties either all terminate $\mathcal{RB}_\mathsf{ok}$ or all not terminate it, and in the former case they all receive the sender's input. An implication of this protocol is that Locher's communication complexity lower bound of $(1.5 - o(1))Ln$ bits for $3$-round reliable broadcast against $t < \frac{n}{3}$ crash faults when the sender is only allowed to send $o(Ln)$ bits in the first round\footnote{The lower bound requires that when the sender is correct, a synchronous implementation of the protocol requires at most $3$ rounds for the correct parties to output. The protocol we gave satisfies this guarantee.} \cite{locher24} no longer holds if the last restriction is lifted.

\subparagraph{\mathversion{bold}The $1.5Ln$ Bound.} We showed in this work by designing $\PullCast$ that for any fixed $\varepsilon > 0$, one can solve non-terminating reliable broadcast with $(1 + \varepsilon)Ln$ bits of communication when the input length $L$ is sufficiently large. This is of interest in light of Locher's $O(1.5 - o(1))Ln$ bit lower bound \cite{locher24}, which he conjectured might extend beyond the class of reliable broadcast protocols for which he proved the bound. As we explained in the previous paragraph, the bound does not extend to all crash fault tolerant reliable broadcast protocols, and $\PullCast$ shows that non-terminating protocols can beat this bound even against stronger adversaries. So, the question is: To what protocols does this bound generalize? Based on the research we have done since this paper's acceptance, we have come to believe that the lower bound holds for all well-balanced reliable broadcast protocols (e.g.\ Locher and Shoup's MiniCast \cite{ls25}), where the expected number of bits sent by each party is roughly equal across all $n$ parties. We intend to formally address this question in future work, as there are subtleties regarding termination and whether crash faults suffice to induce the lower bound.

\bibliography{refs}

\appendix

\section{Information Theory Lemmas} \label{apxsec}

In this section, we prove the information-theoretic lemmas we use to obtain our lower bounds. We begin with Lemma \ref{termlemma}, which is about how much communication a party $i$ must receive from a set of parties $S$ to learn the output $X$ of an asynchronous protocol from them.

\termlemma

\begin{proof}
    Firstly, we assume $|\mathsf{supp}(X)| \geq 2$. If not, then $H(X) = 0$, and the lemma is trivial.

    Let $s = H(X)$. Since $X$ is uniformly random, we have $s = \log_2(|\mathsf{supp}(X)|) \geq 1$. Hence, the fact that $\Pr[Y_i = X] \geq 0.5 + \varepsilon$ gives us $H(X \mid Y_i) \leq H_b(0.5 - \varepsilon) + (0.5 - \varepsilon)\log_2(2^s - 1)$ by Fano's inequality \cite{cover2006elements}, where $H_b(0.5 - \varepsilon) < 1$ is the binary entropy of $(0.5 - \varepsilon)$. Therefore, the mutual information $I(X \mathbin; Y_i)$ is at least $H(X) - H(X \mid Y_i) \geq s - H_b(0.5 - \varepsilon) - (0.5 - \varepsilon)\log_2(2^s - 1) = \Omega(s)$. This function is positive for all $s \geq 1$ because it equals $1 - H_b(0.5 - \varepsilon) > 0$ when $s = 1$ and has a positive derivative in $[1, \infty)$, which means that it is lower bounded by $\varepsilon's = \varepsilon'H(X)$ for some absolute constant $\varepsilon' > 0$.

Observe that $X \rightarrow (T_{S \rightarrow i}, W_i) \rightarrow Y_i$ is a Markov chain, where the variable $W_i$ represents $i$'s initial state (its input if it has any, and its local randomness). The reason is that when $i$ decides to output from the protocol, it computes its output $Y_i$ deterministically based on $T_{S \rightarrow i}$ and $W_i$. So, the data processing inequality \cite{cover2006elements} gives us $I(X \mathbin; T_{S \rightarrow i}, W_i) \geq I(X \mathbin; Y_i) \geq \varepsilon'H(X).$ Moreover, because $X$ and $W_i$ are independent, the chain rule for mutual information gives us $I(X \mathbin; T_{S \rightarrow i}, W_i) = I(X \mathbin; W_i) + I(X \mathbin; T_{S \rightarrow i} \mid W_i) = I(X \mathbin; T_{S \rightarrow i} \mid W_i)$. Consequently, we have $H(T_{S \rightarrow i}) \geq H(T_{S \rightarrow i} \mid W_i) \geq I(X \mathbin; T_{S \rightarrow i} \mid W_i) \geq \varepsilon'H(X)$, as desired.
\end{proof}

Next, we have the lemmas regarding the conversion from transcript entropy to expected communication complexity. They are as follows.

\begin{lemma} \label{sync-cov-lemma}
    Let $i$ be any party, and let $S$ be any non-empty set of parties excluding $i$. If a synchronous $R$-round protocol is such that in its random fault-free executions the bidirectional transcript $T_{S \leftrightarrow i}$ which describes the communication between $S$ and $i$ has at least $\varepsilon R|S|$ bits of entropy for some constant $\varepsilon > 0$, then for some constant $\varepsilon' > 0$ these executions involve at least $\varepsilon'H(T_{S \leftrightarrow i})$ bits of expected communication between $S$ and $i$.
\end{lemma}

\begin{lemma} \label{async-cov-lemma}
    Let $i$ be any party, and let $S$ be any non-empty set of parties excluding $i$. If an asynchronous protocol is such that in its random fault-free executions the one-way transcript $T_{S \rightarrow i}$ which describes the communication from $S$ to $i$ has at least $\varepsilon|S|$ bits of entropy for some constant $\varepsilon > 0$, then in these executions the parties in $S$ send at least $\varepsilon'H(T_{S \rightarrow i})$ bits to $i$ in expectation for some constant $\varepsilon' > 0$, assuming the adversary is such that for any two parties \linebreak $j < j' \in S$ it never delivers a message from $j$ to $i$ after delivering it a message from $j'$.
\end{lemma}

Note the restriction to the adversary's scheduling strategy in Lemma \ref{async-cov-lemma}. This restriction will allow us to encode transcripts more efficiently. We can assume this restriction since in our asynchronous protocol lower bound proofs, the adversary always causes the parties in $S$ to terminate before it delivers any of their messages to $i$. Hence, the adversary can abide by the restriction by sorting the messages before delivering them, and the parties in $S$ will not later on break the assumption by sending more messages.

To prove these lemmas, we first prove a technical proposition about the relation between the entropy of a random string and its expected length.

\begin{proposition}\label{bitlengthlemma}
    For any random bitstring $M$ of a random length $|M|\geq 0$, if $H(M) \geq \varepsilon$ for some constant $\varepsilon > 0$, then $\E[|M|] \geq \varepsilon'H(M)$ for some constant $\varepsilon' > 0$.
\end{proposition}

\begin{proof}
    Firstly, note that $\E[|M|] > 0$. Else, $M$ would be the empty string, with $H(M) = 0$.
    
    Let $L = |M|$ (a random variable itself), let $\supp(L) = \{\ell \in \mathbb{N}:\Pr[L = \ell] > 0\}$, and let $\mu = \E[L] > 0$. For any $\ell \in \supp(L)$, we have $H(M\,|\,L = \ell) \leq \ell$ as the maximum-entropy way to choose an $\ell$-bit string is to choose it uniformly at random, with $\ell$ bits of entropy. Hence, $H(M\,|\,L) = \sum_{\ell \in \supp(L)}\Pr[L = \ell] \cdot H(M\,|\,L = \ell) \leq \sum_{\ell \in \supp(L)}\Pr[L = \ell] \cdot \ell = \mu$. Finally, the chain rule gives us $H(M) \leq H(M, L) = H(M\,|\,L) + H(L) \leq \mu + H(L)$.

    The maximum-entropy way to randomly choose $L \in \{0,1,\dots\}$ given its expected value $\mu$ is to choose $L$ according to a geometric distribution with $\Pr[L = \ell] = \frac{\mu^\ell}{(\mu + 1)^{\ell + 1}}$ for all $\ell \geq 0$ \cite{kapur1989maximum}, which means that $H(L) \leq -\sum_{\ell = 0}^\infty \frac{\mu^\ell}{(\mu + 1)^{\ell + 1}}\log_2(\frac{\mu^\ell}{(\mu + 1)^{\ell + 1}}) = (\mu + 1)\log_2(\mu + 1) - \mu\log_2 \mu$. Hence, we have $H(M) \leq \mu + H(L) \leq \mu + (\mu + 1)\log_2(\mu + 1) - \mu\log_2 \mu$.

    Let $f(x) = x + (x+1)\log_2(x + 1) - x\log_2 x$, and let $g(x) = \frac{x}{f(x)}$. What we would like to show is that $H(M) \geq \varepsilon \implies \mu \geq \varepsilon'H(M)$. We have $H(M) \leq f(\mu)$, which implies that $\mu = g(\mu) \cdot f(\mu) \geq g(\mu) \cdot H(M)$. So, our goal is to find some $\varepsilon' > 0$ which depends only on $\varepsilon$ such that $g(\mu) \geq \varepsilon'$. As $f(x)$ is continuously and strictly increasing in $(0, \infty)$ without an upper bound and $\lim_{x\rightarrow 0^+}f(x) = 0$, $f$ has an inverse $f^{-1} : (0, \infty) \rightarrow (0, \infty)$. Let $\mu_0 = f^{-1}(\varepsilon)$. \linebreak We have $f(\mu) \geq H(M) \geq \varepsilon = f(\mu_0)$, which means that $\mu \geq \mu_0$ because $f$ is increasing. We can thus set $\varepsilon' = g(\mu_0) = \frac{\mu_0}{\varepsilon}$ and conclude that $\mu \geq \varepsilon'H(M)$. This choice works because $g$ is also an increasing function in $(0, \infty)$, which means that $g(\mu) \geq g(\mu_0) = \varepsilon'$.

    Note that the $\varepsilon' = \frac{f^{-1}(\varepsilon)}{\varepsilon}$ we get from this argument approaches $1$ as $\varepsilon \rightarrow \infty$, due to the fact that $f(x) \approx x \approx f^{-1}(x)$ for large values of $x$. 
\end{proof}

We use Proposition \ref{bitlengthlemma} to prove a generalization of it for arrays of random bitstrings.

\begin{lemma} \label{arraylemma}
    For any $N \geq 1$, let $M_1, \dots, M_N$ be any sequence of $N$ variable-length random bitstrings (which are possibly empty) such that $\sum_{i = 1}^NH(M_i) \geq \varepsilon N$ for some constant $\varepsilon > 0$. Then $\E[\sum_{i = 1}^N|M_i|] \geq \varepsilon'\sum_{i = 1}^NH(M_i)$ for some constant $\varepsilon' > 0$.
\end{lemma}

\begin{proof}
    Let $I_{\mathsf{small}} = \{i \in [N] : H(M_i) < \frac{\varepsilon}{2}\}$, and let $I_{\mathsf{large}} = [N] \setminus I_{\mathsf{small}}$. Observe that we have $\sum_{i \in I_{\mathsf{small}}}H(M_i) \leq \frac{\varepsilon N}{2}$, which implies $\sum_{i \in I_{\mathsf{large}}}H(M_i) \geq \sum_{i = 1}^NH(M_i) - \frac{\varepsilon N}{2} \geq \frac{1}{2}\sum_{i = 1}^NH(M_i)$. By Proposition \ref{bitlengthlemma}, there is a constant $\varepsilon' > 0$ such that $\E[|M_i|] \geq \varepsilon'H(M_i)$ for all $i \in I_{\mathsf{large}}$, which gives us $\E[\sum_{i = 1}^N|M_i|] \geq \E[\sum_{i \in I_{\mathsf{large}}}|M_i|] \geq \varepsilon'\sum_{i \in I_{\mathsf{large}}}H(M_i) \geq \frac{\varepsilon'}{2}\sum_{i = 1}^NH(M_i)$.
\end{proof}

Finally, to prove Lemmas \ref{sync-cov-lemma} and \ref{async-cov-lemma}, we assume that the parties' messages are encoded with a prefix-free code, so that the concatenation $m_1 || m_2 || \dots || m_k$ of any messages can be decoded into the vector $(m_1, \dots, m_k)$. We can assume this without loss of generality since we \linebreak can get prefix-free messages without an asymptotic cost via Elias delta encoding \cite{elias}.

\begin{proof}[Proof of Lemma \ref{sync-cov-lemma}.]
    Suppose $S = \{s_1, \dots, s_{|S|}\}$. As the protocol takes $R$ rounds, we can represent $T_{S \leftrightarrow i}$ with an $R \times |S| \times 2$ multidimensional array $A$ which is defined as follows: For \linebreak all $(r, j) \in [R] \times [|S|]$, $A[r][j][0]$ is filled with the concatenation of the messages $i$ received in order from $s_j$ on round $r$ (an empty string if $i$ received no messages from $s_j$), and likewise $A[r][j][1]$ is filled with the concatenation of the messages $s_j$ received from $i$ on round $r$.

    Suppose we encode the transcript $T_{S \leftrightarrow i}$ into an array $A$ via the encoding described above. For each $c = (c_1, c_2, c_3) \in [R] \times [|S|] \times [2]$, let $M_c$ refer to random string in the cell $A[c_1][c_2][c_3]$. We want to lower bound the expected value of $B = \sum_{c}|M_c|$, as $B$ is the number of bits sent. For this, we observe that the transcript $T_{S \leftrightarrow i}$ is fully described by the cells $M_c$, which means that we have $\sum_{c}H(M_c) \geq H(T_{S \leftrightarrow i}) \geq \varepsilon R|S|$ by the sub-additivity of entropy. By Lemma \ref{arraylemma}, \linebreak this gives us $\E[B] = \E[\sum_{c}|M_c|] \geq \varepsilon'\sum_{c}H(M_c) \geq \varepsilon'H(T_{S \leftrightarrow i})$ for some constant $\varepsilon' > 0$.
\end{proof}

\begin{proof}[Proof of Lemma \ref{async-cov-lemma}.]
    Suppose $S = \{s_1, \dots, s_{|S|}\}$. We encode the transcript $T_{S \rightarrow i}$ into an $|S|$-element array $A$ as follows: For each $j \in [|S|]$, we fill the cell $A[j]$ with the concatenation of the messages $i$ received from $s_jj$, in the order $i$ received them. Under the ordering assumption we made about how the adversary schedules messages, we know for all $j < j'$ that $i$ receives all of the messages represented by $A[j]$ from $s_j$ before receiving any of the messages represented by $A[j']$ from $s_{j'}$; hence, under this scheduling assumption, the array represents the transcript without any information loss.
    
    We want to lower bound the expected value of $B = \sum_{j = 1}^{|S|}|A[j]|$, since this is the number of bits that $i$ receives from $S$. As the transcript $T_{S \rightarrow i}$ is fully described by the array $A$, we have $\sum_{j=1}^{|S|}H(A[j]) \geq H(T_{S \rightarrow i}) \geq \varepsilon |S|$ by the sub-additivity of entropy. Finally, by Lemma \ref{arraylemma}, we have $\E[B] = \E[\sum_{j = 1}^{|S|}|A[j]|] \geq \varepsilon'\sum_{j=1}^{|S|}H(A[j]) \geq \varepsilon'H(T_{S \rightarrow i})$ for some constant $\varepsilon' > 0$.
\end{proof}

\section{Bounds for All Sufficiently Large \texorpdfstring{\mathversion{bold}$n$}{n}} \label{infinitely-many}

In Sections \ref{syncsection} and \ref{asyncsection}, we prove communication complexity lower bounds that hold assuming certain adversary structures exist. That is, we prove that whenever $n$ and $d$ are such that the appropriate adversary structure $\Z_\pj^{n,d}$ or $\Z_\tpj^{n,d}$ exists, the communication complexity of a task is lower bounded by $\Omega(Ln^{1 + 1/d})$ or $\Omega(Ln^{2 + 1/d})$. In this section, we extend our bounds by showing that they hold for all sufficiently large $n$, if $d$ is any fixed constant. We begin by showing this for the bounds that use the adversary structure family $\Z_\pj^{n,d}$.

Let $p_1, p_2, p_3, \dots$ be the list of all prime powers in ascending order, let $q(n,d) \geq 1$ be the greatest integer such that $\sum_{k=0}^d(p_{q(n,d)})^k = \frac{(p_{q(n,d)})^{d+1} - 1}{p_{q(n,d)}-1} \leq \frac{n}{2}$, and let $s_{n,d} = \frac{(p_{q(n,d)})^{d+1} - 1}{p_{q(n,d)}-1}$. The adversary structure $\Z_\pj^{n,d}$ exists for any $d \geq 1$ if and only if $n \geq 2^{d+2} - 2$ (so that $q(n,d)$ is well-defined) and $n = 2s_{n,d}$. Recall that in Section \ref{advstructures} we construct $\Z_\pj^{n,d}$ by taking the union of two party sets $U$ and $U'$ of size $s_{n,d}$, where the adversary can always corrupt any party in $U'$. To extend the construction so that it works for all $n \geq 2^{d+2} - 2$, we could redefine $U'$ so that it is a party set of size $n - s_{n,d}$ instead of $s_{n,d}$ while keeping the adversary's ability to corrupt any party in $U'$. Our lower bounds hinge on the fact that $|U'| = \Theta(n)$, which still holds since we will now have $|U'| = n - 2s_{n,d} \geq \frac{n}{2}$, and on the fact that $|U| = s_{n,d} = \Omega(n)$ (with this implying that each party in $U$ belongs to a $\Theta(|U|^{-1/d}) = \Theta(n^{-1/d})$ fraction of the projective geometry quorums), which is no longer guaranteed. Below, we show that there exists some $N_d \geq 2^{d+2} - 2$ such that $n \geq N_d \implies |U| = s_{n,d} > \frac{n}{4}$, which means that the guarantees we want hold for all $n \geq N_d$.

Observe that for any fixed $d$, we have $\lim_{n \rightarrow \infty}q(n,d) = \infty$. So, by the prime number theorem \cite{Yan2000}, which implies that the ratio of successive primes (and thus the ratio of successive prime powers) tends to $1$, we have $\lim_{n \rightarrow \infty}\frac{p_{q(n,d) + 1}}{p_{q(n,d)}} = 1$. This means that there exists some $N_d \geq 2^{d+2} - 2$ such that for all $n \geq N_d$ we have $\frac{p_{q(n,d) + 1}}{p_{q(n,d)}} \leq \sqrt[d+1]{2}$. For all $n \geq N_d$, this gives us $2s_{n,d} = \frac{2(p_{q(n,d)})^{d+1} - 2}{p_{q(n,d)}-1} \geq \frac{(p_{q(n,d) + 1})^{d+1} - 2}{p_{q(n,d)}-1} > \frac{(p_{q(n,d) + 1})^{d+1} - 1}{p_{q(n,d)}} \geq \frac{(p_{q(n,d) + 1})^{d+1} - 1}{p_{q(n,d) + 1} - 1} > \frac{n}{2}$, where the last inequality holds due to $q(n,d)$ being the greatest integer such that $\frac{(p_{q(n,d)})^{d+1} - 1}{p_{q(n,d)}-1} \leq \frac{n}{2}$. In conclusion, we have $n \geq N_d \implies s_{n,d} > \frac{n}{4}$, as desired.

The calculations for $\Z_\tpj^{n,d}$ (which our core set agreement lower bound uses) can be done analogously, by redefining $q(n,d) \geq 1$ to be the greatest integer such that $\sum_{k=0}^d(p_{q(n,d)})^k = \frac{(p_{q(n,d)})^{d+1} - 1}{p_{q(n,d)}-1} \leq \frac{n}{3}$ (well-defined when $n \geq 3(2^{d+1} - 1)$), recalling that $\Z_\tpj^{n,d}$ partitions the parties into three sets $U_1, U_2, U'$ of size $s_{n,d} = \frac{(p_{q(n,d)})^{d+1} - 1}{p_{q(n,d)}-1}$ each, increasing the number of parties in $U'$ from $s_{n,d}$ to $n - 2s_{n,d}$ to make the construction work whenever $n \geq 3(2^{d+1} - 1)$, and using the same calculations as above to show that there is some $N_d \geq 3(2^{d+1} - 1)$ such that $n \geq N_d \implies 2s_{n,d} > \frac{n}{3}$. When $n \geq N_d$, we would get $|U_1| = |U_2| > \frac{n}{6}$ and $|U'| \geq \frac{n}{3}$, \linebreak and these suffice for our core set agreement lower bound (Theorem \ref{acsthm}) to hold.

\end{document}